\documentclass[12pt]{article}
\usepackage{amsmath}
\usepackage{graphicx,psfrag,epsf}
\usepackage{enumerate}
\usepackage{psfrag,epsf}
\usepackage{url} 
\usepackage{amsmath,amssymb,amsfonts,amsthm,mathtools,mathrsfs}

\usepackage{bm,bbm}
\usepackage{color}
\usepackage{subfig}
\usepackage{algorithm,algpseudocode}
\usepackage[symbol]{footmisc}
\usepackage{scalefnt}
\usepackage{authblk} 
\usepackage{multirow,centernot}
\usepackage[colorlinks=true, citecolor=blue, urlcolor=blue]{hyperref}
\usepackage{natbib}
\usepackage{dsfont}
\usepackage[title]{appendix}
\input cyracc.def

\usepackage[left=1in,top=1.1in,right=0.5in,bottom=1in]{geometry}

\theoremstyle{definition}

\newtheorem{theorem}{Theorem}[section]

\newtheorem{proposition}[theorem]{Proposition}
\newtheorem{remark}[theorem]{Remark}

\makeatletter
\def\@seccntformat#1{\@ifundefined{#1@cntformat}%
	{\csname the#1\endcsname\quad}
	{\csname #1@cntformat\endcsname}
}
\makeatother

\newif\ifShowComments
\ShowCommentstrue
\def\strutdepth{\dp\strutbox}
\def\druk#1{\strut\vadjust{\kern-\strutdepth
        {\vtop to \strutdepth{%
                \baselineskip\strutdepth\vss
                        \llap{\hbox{#1}\quad}\null}}}}

\title{\bf
A novel unit-asymmetric 
distribution based on correlated Fréchet random variables
}

\author{
\text{Roberto Vila}$^{1}$\thanks{Corresponding author: Roberto Vila, email: {rovig161@gmail.com}
\newline
}\,\,\,and
\text{Felipe Quintino}$^{1}$\\
{\small $^{1}$ Department of Statistics, University of Brasilia, Brasilia, Brazil}
}
\date{}                     
\begin{document}
	\maketitle 	
	\begin{abstract}
In this paper, we propose a new distribution with unitary support which can be characterized as a ratio of the type $W=X_1/(X_1+X_2)$, where $(X_1, X_2)^\top$ follows a bivariate extreme distribution with Fréchet margins, that is, $X_1$ and $X_2$ are two correlated Fréchet random variables. Some mathematical properties such as identifiability, symmetry, stochastic representation, characterization as a ratio, moments, stress-strength probability, quantiles, and the maximum likelihood method are rigorously analyzed. Two applications of the ratio distribution are discussed.	\end{abstract}
	
	
	\section{Introduction}
    
	\noindent

Probability models capable of handling data with unit support have garnered considerable attention in statistical research in recent years. The limitations on the support of real-world data naturally arise depending on the phenomena studied. One method to develop new unit models is by applying suitable transformations to already-known models. For example, $Z=X/(X+1)$, $Z=\exp(-X)$ and $Z=X/(X+Y)$ have the support $(0, 1)$, when $X$ and $Y$ are positive variables.
Applications of unit models are found in several areas.
\cite{Ghitanya19} proposed the unit-inverse Gaussian distribution and applied this new model to the Municipal Human Development Index of some Brazilian cities. 
\cite{Bourguignon24} reparameterized the three-parameter Generalized Beta distribution and proposed a new parametric quantile model to analyze a real COVID-19 dataset from Chile.
\cite{Martínez-Flórez24} introduced the unit-bimodal Birnbaum-Saunders distribution. They modeled the distribution of the body fat percentage of Australian athletes competing in different sports.
\cite{Vila24} proposed a class of unit‑log‑symmetric models to analyzing an internet access data.

Recently introduced unit models also include the power unit inverse Lindley distribution \citep{Gemeay24}, the unit extended exponential distribution \citep{Ragab24}, the log-cosine-power unit distribution \citep{Nasiru24}, the Gumbel–Logistic unit distribution \citep{Stojanović24}, and the bimodal beta distribution \citep{Vila24bimodalrates}.
The literature has widely studied models based on Birnbaum-Saunders (BS) distribution, and unit-BS models were also proposed \citep[cf.][]{Mazucheli18, Martínez-Flórez24, Vila24unitbs}. 

New regression models for data on the unit interval have garnered significant attention from researchers in the past few years  \citep[cf.][]{Benites22, Hussein23, Karakaya24, Mazucheli21quantilereg, Santoro24, Stojanović2024Laplace}.
Another important class of models to highlight is the multivariate models for unit cylinder data, where the dependency between random variables can be formally described \citep[cf.][]{Martínez-Flórez22, Vila24bivariate, Vila24multivariate}.

As previously mentioned, models for unit data can naturally be constructed from ratios of random variables. Along these lines, \cite{Bekker09} proposed the type I distribution of the ratio of independent Weibullized generalized beta-prime variables, while 
\cite{NadarajahKotz06} studied the ratio of independent Fréchet random variables.

In this work, we introduce a new class of unit models that arises from a ratio of two correlated Fréchet random variables. 
Recall that the Fréchet distribution is an extreme value model used for modeling extreme events related to maximums of independent random samples.
Applications of this distribution are found in finance, natural catastrophes, and equipment failures, among others.
The books \cite{Elshamy1992, EmbrechtsKluppelbergMikosch13, Galambos78, HaanFerreira06, Resnick08} provide extensive coverage for a detailed study of extreme value theory.

In bivariate data modeling, it is common to focus on comparing both random variables. For instance, in Finance, we may compare two assets \citep[cf.][]{Quintino23}, while in Engineering, we might examine the relationship between a material's strength and the stress it is subjected to \citep[cf.][]{KunduRaqab09}.
Several methods can be used to make these comparisons. One such method is based on stress-strength reliability (SSR). Given two random variables, $X$ and $Y$, we can compare them by analyzing the probability that $\mathbb{P}(X < Y)$ is greater or less than $1/2$. \cite{Nadarajah03} and \cite{AbbasTang14} studied SSR for independent Fréchet random variables $X$ and $Y$.
An alternative approach is to analyze the ratio 
$X/Y$ directly. In this context, \cite{NadarajahKotz06} examined the ratio of independent Fréchet random variables. They derived the cumulative distribution function (CDF) and the percentiles of the ratio using special functions such as the error function and the generalized hypergeometric function.

It also is possible to compare $X$ and $Y$ when both are dependent random variables. In this context, we can model the dependence between the random variables via a copula approach. Several applications of bivariate copula models for extreme value distributions are presented in the Literature \citep[see, for example,][and the references therein]{Lima24}. 
Furthermore, asymptotic constructions, characterizations, and various properties of multivariate extreme value distributions are extensively discussed in \cite{Elshamy1992} and \cite{Resnick08}.
In particular, a class of bivariate distributions with Fréchet marginals presented in \cite{Elshamy1992} served as the basis for constructing the new distribution with unit support proposed in this work. A similar approach could be applied to the bivariate classes with Fréchet marginals presented in \cite{Resnick08}.

This paper aims to propose a new class of distributions with unit support. Our construction generalizes the ratio model proposed by \cite{NadarajahKotz06}, which was based on the ratio of independent Fréchet random variables. A Monte Carlo simulation study is presented to validate the maximum likelihood estimator of the new distribution. Finally, applications to real-world datasets are discussed.

The remainder of this paper is organized as follows: Section \ref{sec_preliminaries} provides an introduction to the preliminaries, including the definition of the univariate Fréchet distribution and a bivariate distribution with Fréchet marginals. Section \ref{sec_main} introduces the new model, referred to as the unit-Fréchet distribution.
Section \ref{sec_properties}, on the other hand, deals with the
derivation of the new unit distribution and its mathematical properties such as identifiability, symmetry, stochastic representation, characterization as a ratio, moments, stress-strength probability, quantiles, and the maximum likelihood estimation for the parameters.
In Section \ref{sec_simulations}, we discuss Monte Carlo simulations to evaluate the proposed estimator, while Section \ref{sec_applications} addresses the modeling of two real-world scenarios involving football data and income-consumption data set. The final section presents our conclusions.

\section{Preliminaries}\label{sec_preliminaries}
In this section, we present some definitions which will be used subsequently.

\subsection{The univariate Fréchet distribution}
Let $X_1, X_2, \cdots$ be a sequence of independent and identically distributed random variables with a common distribution $G(\cdot)$. 
Consider the partial maximum $X_{(n)} = \max \{X_j;~ j=1,\cdots, n\}$.
The derivation of extreme distributions (asymptotic distribution functions for $X_{(n)}$) is based on the asymptotic approximation of the normalized partial maximum.
More specifically,
\begin{equation}\label{eq_partial_maximum}
    \frac{X_{(n)}-b_n}{a_n}\overset{d}{\rightarrow}F(x),~~\mbox{as}~~n\rightarrow\infty,
\end{equation}
where $a_n>0$, $b_n\in\mathbb{R}$, $F(\cdot)$ is a nondegenerate distribution function and $\overset{d}{\rightarrow}$ denotes convergence in distribution.

Fisher and Tippet Theorem \citep[cf.][]{Resnick08} guarantees us that there are three extreme values distributions $F(\cdot)$ satisfying (\ref{eq_partial_maximum}), the Fréchet, reversed Weibull, and Gumbel distributions.
The conditions on the distribution function $G(\cdot)$ for (\ref{eq_partial_maximum}) to hold true are well-documented and can be referenced in \cite{Resnick08, Galambos78}.
Applications of these distributions are found in finance, natural catastrophes, and equipment failures, among others  \cite{EmbrechtsKluppelbergMikosch13}.

A random variable $X$ has Fréchet distribution with location parameter $\mu\in\mathbb{R}$, scale parameter $\sigma>0$, and shape parameter $\alpha>0$, denoted by $X\sim$ Fréchet$(\mu, \sigma, \alpha)$, if its cumulative distribution function (CDF) and probability density function (PDF), respectively, are given by
$$F(x; \mu, \sigma, \alpha) = \left\{ \begin{array}{cc}
   0,  & x<\mu, \\
  \exp\left\{ -\left( \frac{x-\mu}{\sigma} \right)^{-\alpha} \right\},    & x\geq \mu,
\end{array}  \right.$$
and
\begin{equation*}\label{eq_PDF_Frechet}
    f(x; \mu, \sigma, \alpha) = \frac{\alpha}{\sigma}\left( \frac{x-\mu}{\sigma} \right)^{-\alpha-1}  F(x; \mu, \sigma, \alpha)\mathbbm{1}_{(\mu, \infty)}(x),
\end{equation*}
where $\mathbbm{1}_A$ denotes the indicator function on the set $A$.
Figure \ref{fig:density_Frechet} shows the behavior of $f(x; \mu, \sigma, \alpha)$ for some parameter choices.

\begin{figure}[htb!]
	\centering
	\includegraphics[width=1.0\linewidth]{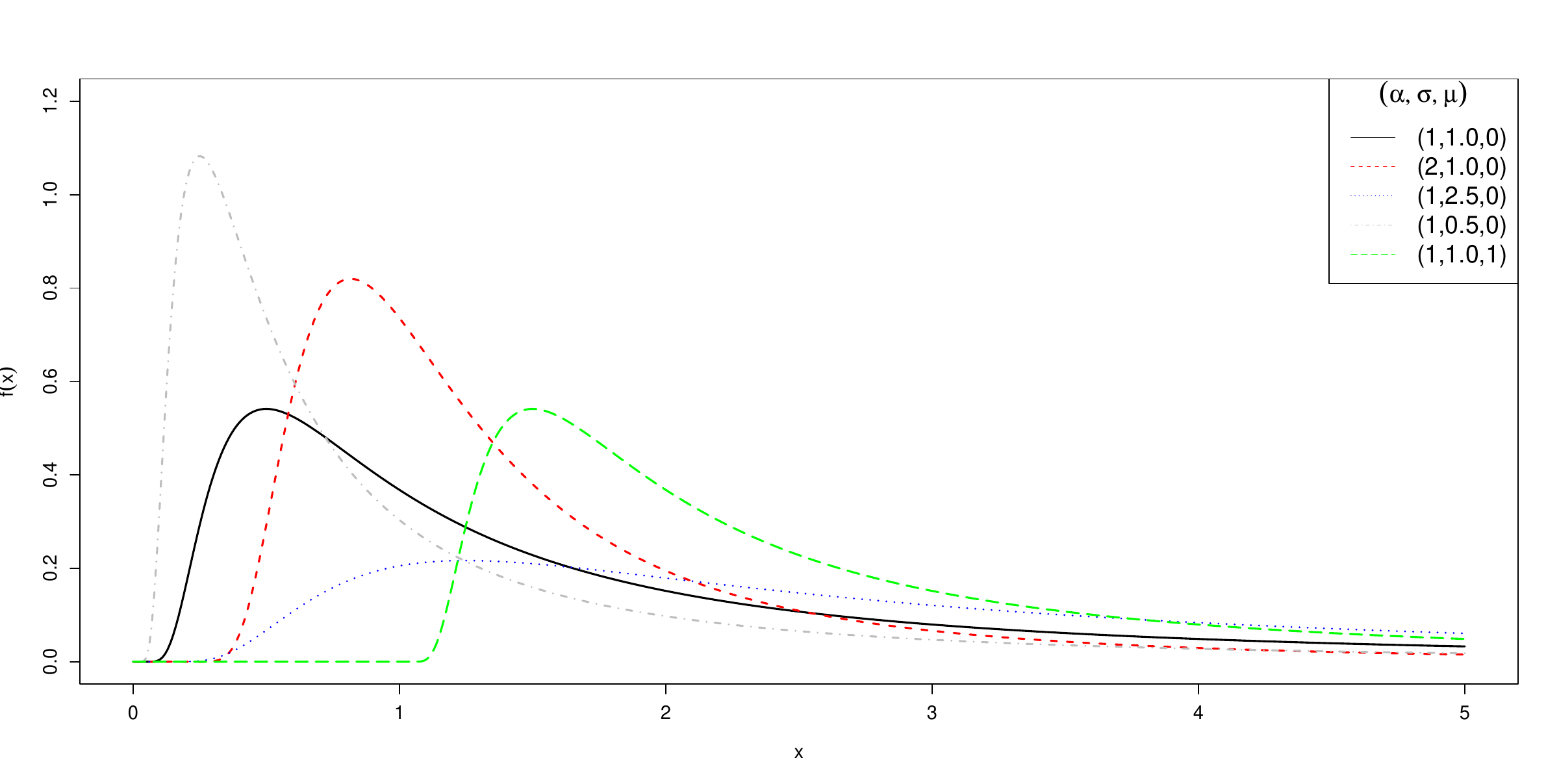}
	\caption{Plot of the PDF $f(x; \mu, \sigma, \alpha)$ with varying parameters.}
	\label{fig:density_Frechet}
\end{figure}

\subsection{A bivariate distribution with Fréchet marginals}
 Let $(X_1, Y_1), (X_2, Y_2), \cdots $ be a sequence of independent and identically distributed bivariate random vectors with the same joint CDF $G(x,y)$. We consider the pair
 $$(X_{(n)}, Y_{(n)}) := \left( \max \{X_j;~ j=1,\cdots, n\}, \max \{Y_j;~ j=1,\cdots, n\} \right).$$
Analogously to the univariate case, we are interested in the (nondegenerate) limit distribution satisfying
  \begin{equation}\label{eq_partial_maximum_bivariate}
    \left(\frac{X_{(n)}-b_n}{a_n}, \frac{Y_{(n)}-d_n}{c_n}\right)\overset{d}{\rightarrow}F(x,y),~~\mbox{as}~~n\rightarrow\infty,
\end{equation}  
where $a_n>0$ and $c_n>0$. If a joint CDF $F(x,y)$ satisfies \eqref{eq_partial_maximum_bivariate}, we call it a bivariate extreme value distribution.

Two general forms for the bivariate extreme value distributions $F(x,y)$ in terms of the marginal distributions were described in \cite{Gumbel1967} \cite[see also][]{Elshamy1992}.

We say that a random vector  $(X, Y)$ follows a bivariate extreme distribution with Fréchet margins  \cite[as discussed in][p. 14, Item 21]{Elshamy1992} if its joint CDF is given by
	\begin{align}\label{id-00}
    		F_{X,Y}(x,y)
		=
		\exp\left\{-\left({x\over\sigma_1}\right)^{-\alpha}-\left({y\over\sigma_2}\right)^{-\alpha} +
		\rho\left[\left({x\over\sigma_1}\right)^\alpha+\left({y\over\sigma_2}\right)^\alpha\right]^{-1}\right\},
	\end{align}
	where $x\geqslant0$, $y\geqslant 0$;
	$\sigma_1>0,\sigma_2>0, \alpha>0$ and $0\leqslant \rho\leqslant 1$. The parameter $\rho$ represents the degree of association between $X$ and $Y$. A value of $\rho=0$ indicates that $X$ and $Y$ are independent with distributions Fréchet$(0, \sigma_1, \alpha)$ and Fréchet$(0, \sigma_2, \alpha)$, respectively. Figure \ref{fig:jointPDF} shows the behavior of the joint PDF corresponding to \eqref{eq_partial_maximum_bivariate} and the contour lines of the joint density for $\alpha=2, \sigma_1=\sigma_2=1$ and $\rho\in\{0, 0.5, 0.9\}$.

\begin{figure}[htb!]
    \centering
    \subfloat[$\rho=0$]{\includegraphics[width=.5\linewidth]{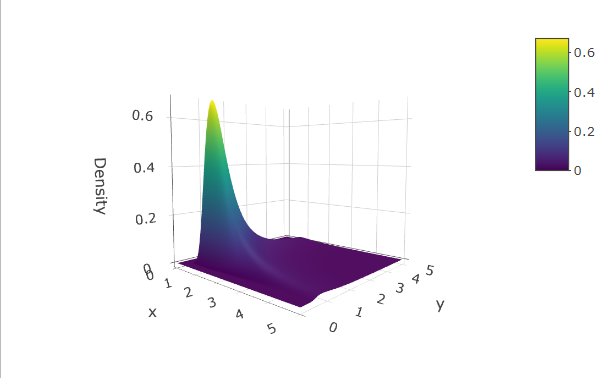}}
    \subfloat[$\rho=0$]{\includegraphics[width=.5\linewidth]{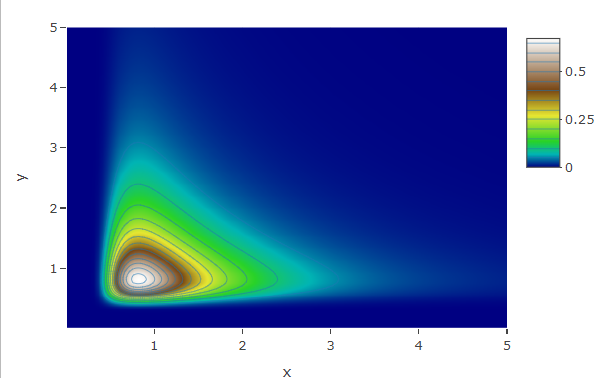}} \\
    \subfloat[$\rho=0.5$]{\includegraphics[width=.5\linewidth]{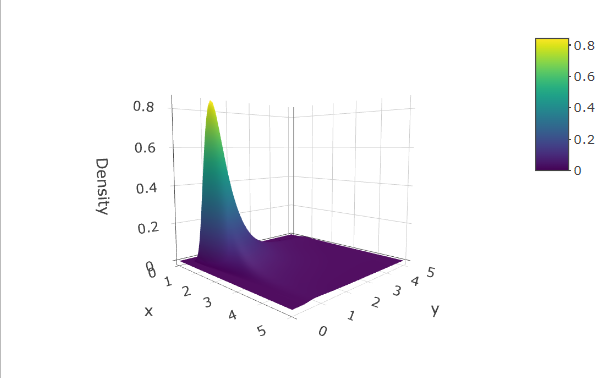}}
    \subfloat[$\rho=0.5$]{\includegraphics[width=.5\linewidth]{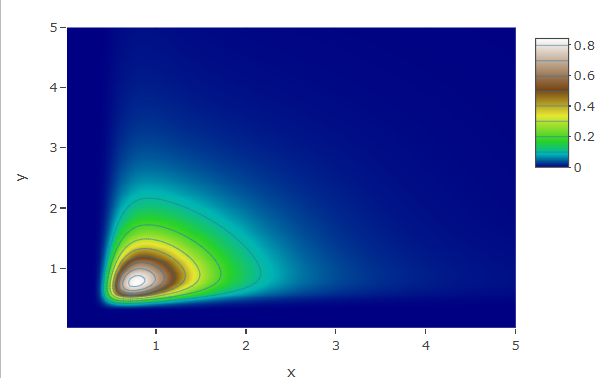}}\\
        \subfloat[$\rho=0.9$]{\includegraphics[width=.5\linewidth]{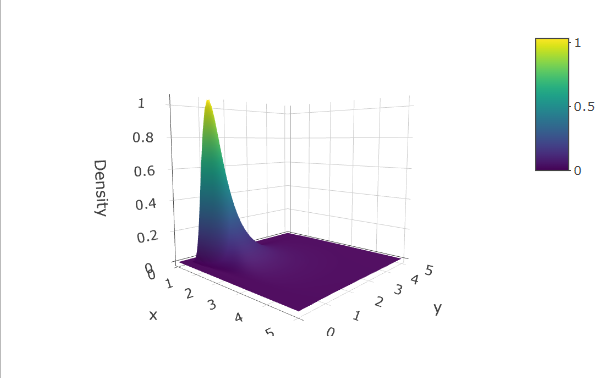}}
    \subfloat[$\rho=0.9$]{\includegraphics[width=.5\linewidth]{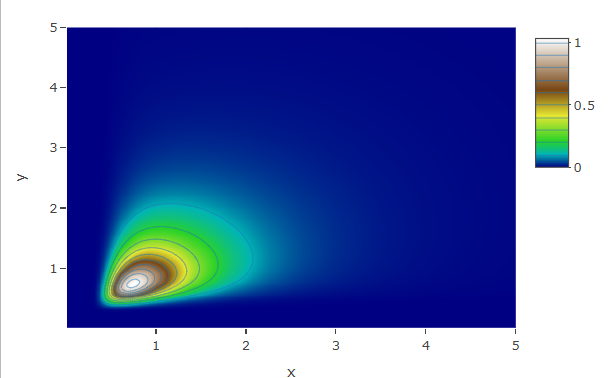}} 
    \caption{Joint PDF of bivariate extreme distribution with Fréchet margins and contour lines of the joint density for $\rho=0$ in (a) and (b), $\rho=0.5$ in (c) and (d), and $\rho=0.9$ in (e) and (f). $\alpha=2, \sigma_1=\sigma_2=1$.}
    \label{fig:jointPDF}
\end{figure}

 \section{The unit Fréchet distribution} \label{sec_main}

	We say that a continuous random variable $W$, with support $(0, 1)$, follows a unit-extreme (UF) distribution with parameter vector $\boldsymbol{\theta} = (\sigma,\alpha,\rho)^\top$, $\sigma > 0$, $\alpha> 0$ and $0\leqslant \rho\leqslant 1$, denoted by $W\sim {\rm UF}(\boldsymbol{\theta})$, if its PDF is given by (for $0 < w < 1$)
	\begin{align}\label{pdf-main}
		f_W(w;\boldsymbol{\theta})
		=
		{\alpha\over \sigma^\alpha}\,
		s^{\alpha-1}(s+1)^2
		\left\{ \displaystyle
		{\displaystyle2\left({s^\alpha\over \sigma^\alpha}+1\right)^2
			-
			\rho\left[\left({s^\alpha\over \sigma^\alpha}\right)^2+1\right]
			\over 
			\displaystyle
			\left[
			\left({s^\alpha\over \sigma^\alpha}+1\right)^2
			-\rho\, {s^\alpha\over\sigma^\alpha}
			\right]^2}
		- 
		{\displaystyle 1\over\displaystyle \left({s^\alpha\over \sigma^\alpha} + 1\right)^2}
		\right\},
		\quad 
		s={w\over 1-w}.
	\end{align}
	
	As 
	\begin{align*}
		{\partial {s^\alpha\over \sigma^\alpha}\over\partial w}
		=
		{\alpha\over \sigma^\alpha}\, s^{\alpha-1}(s+1)^2,
	\end{align*}
	by the chain rule, the PDF of $W$ is written as
	\begin{align}\label{f-expression}
	f_W(w;\boldsymbol{\theta})
	=
	{\partial G({s^\alpha\over \sigma^\alpha};\rho)\over\partial w},
	\end{align}
	where $G(x;\rho)=\int_0^x g(t;\rho){\rm d}t$, $x>0$, is the CDF corresponding to the following PDF
	\begin{align}\label{def-g-cdf}
	g(x;\rho)={2(x+1)^2-\rho(x^2+1)\over [(x+1)^2-\rho x]^2}-{1\over (x+1)^2}, \quad x>0.
	\end{align}
	It is clear that
	\begin{align}\label{def-G-cdf}
	G(x;\rho)
	=
	{x\over x+1}\, {(x+1)^2-\rho\over (x+1)^2-\rho x}, \quad x>0.
	\end{align}
	Therefore, from \eqref{f-expression}, the CDF of $W$ is given by (for $0 < w < 1$)
	\begin{align}\label{CDF}
	F_W(w;\boldsymbol{\theta})
	=
	G\left({s^\alpha\over \sigma^\alpha};\rho\right)
	=
		\dfrac{\displaystyle {s^\alpha\over\sigma^\alpha} }{\displaystyle {s^\alpha\over\sigma^\alpha}+1}
	\,
	\dfrac{\displaystyle\left({s^\alpha\over\sigma^\alpha}+1\right)^{2}-\rho}
	{
		\displaystyle\left({s^\alpha\over\sigma^\alpha}+1\right)^{2}-\rho\,{s^\alpha\over\sigma^\alpha}
	},
	\quad 
	s={w\over 1-w}.
	\end{align}

\begin{remark}
	Taking $\rho=0$ in \eqref{CDF}, we get (for $0 < w < 1$)
	\begin{align*}
	F_W(w;\boldsymbol{\theta})
	=
	\dfrac{\displaystyle {s^\alpha}}{\displaystyle {s^\alpha}+\sigma^\alpha}
	=
		\dfrac{\displaystyle {w^\alpha}}{\displaystyle {w^\alpha}+\sigma^\alpha(1-w)^\alpha}.
	\end{align*}
\end{remark}

Figures \ref{fig:f_W} and \ref{fig:F_W} show the behavior of $f_W$ for some parameter choices.
Note that we can have a symmetric distribution centered around 0.5, with a shape that varies according to the concentration around the mean (including a U shape). It is also possible to obtain asymmetric models, depending on the choice of parameters.

\begin{figure}[htb!]
	\centering
	\includegraphics[width=1.0\linewidth]{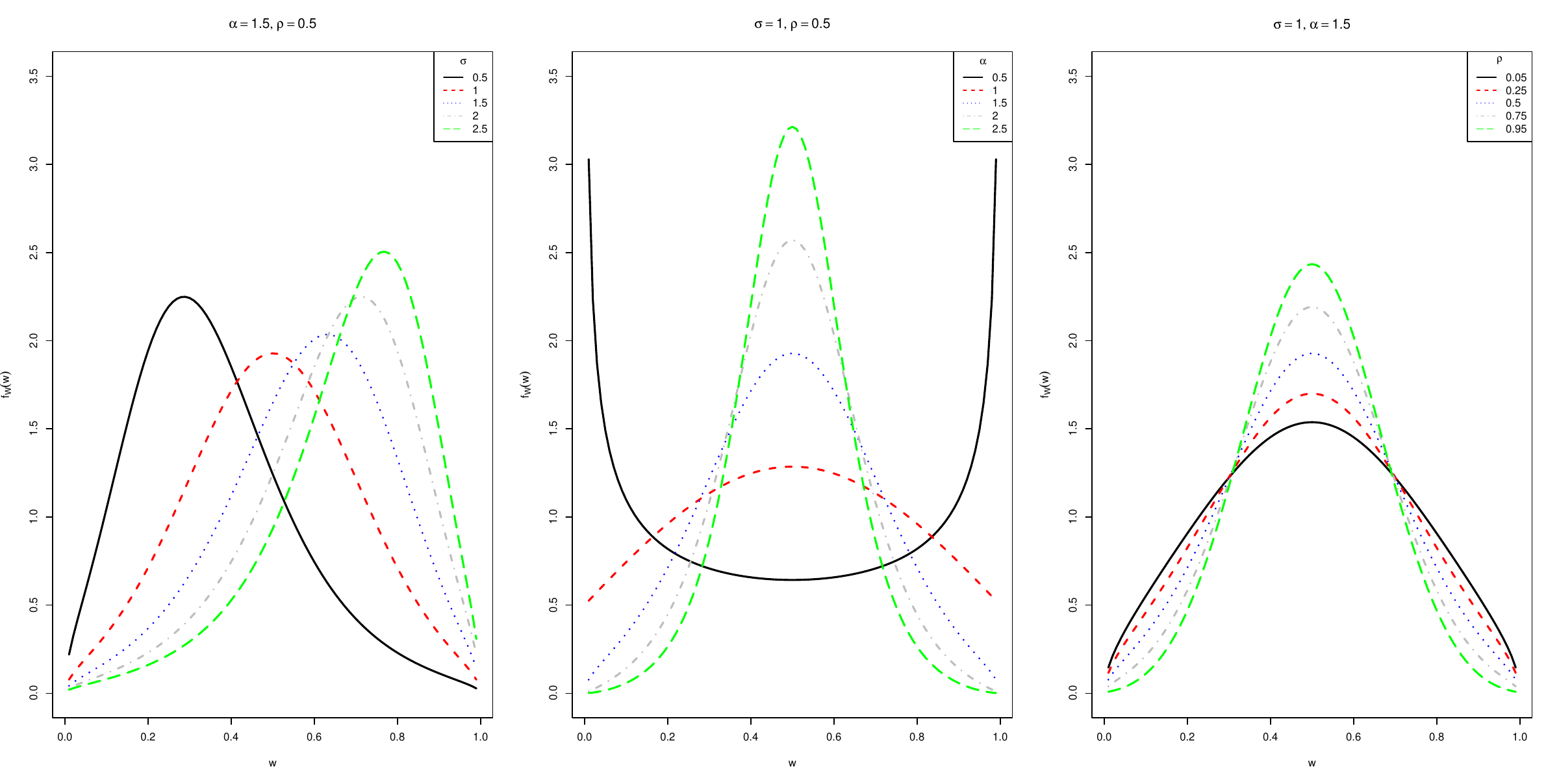}
	\caption{Plot of the PDF $f_W$ with varying parameters $\alpha_1$ (left), $\alpha_2$ (middle) and $\rho$ (right).}
	\label{fig:f_W}
\end{figure}

\begin{figure}[htb!]
	\centering
	\includegraphics[width=1.0\linewidth]{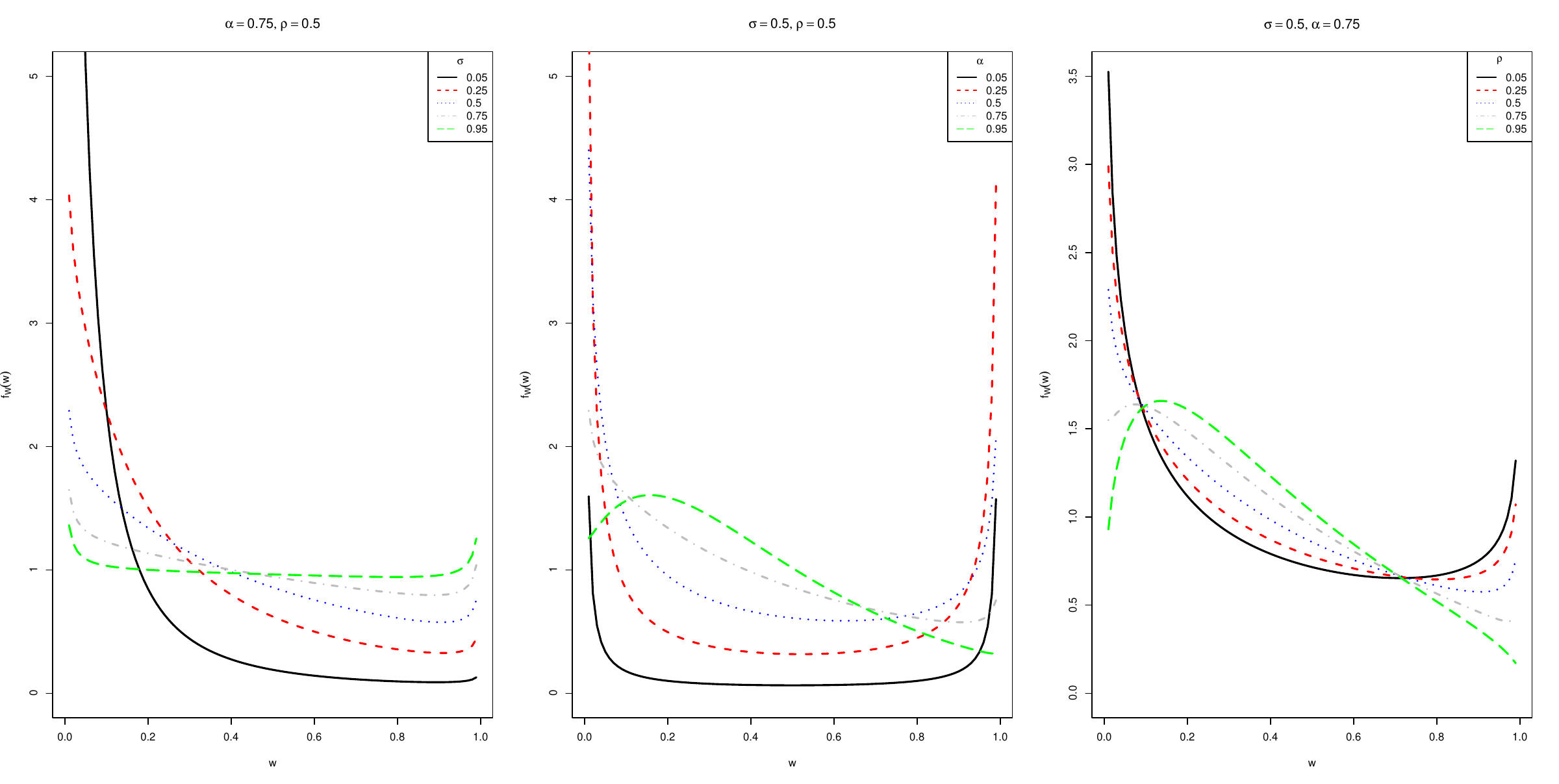}
	\caption{Plot of the PDF $f_W$ with varying parameters $\alpha_1$ (left), $\alpha_2$ (middle) and $\rho$ (right).}
	\label{fig:F_W}
\end{figure}    
	
	\section{Some basic properties}\label{sec_properties}
	We establish here some basic mathematical properties of the unit-Fréchet distribution.
	
%
%

	\subsection{Identifiability}
	
	Under some restrictions on the parameter space $\Theta\equiv(0,\infty)\times (0, \infty)\times [0,1]$, the following result
	establishes the identifiability of UF distribution. To establish this result, we write $\Theta_*=\{\boldsymbol{\theta}\in\Theta: \alpha\in\mathbb{N} \ \text{and} \ \rho<1\}$.
	
	\begin{theorem}[Identifiability] 
The mapping $\Theta_* \ni \boldsymbol{\theta}\mapsto F_W(w;\boldsymbol{\theta})$, $\forall 0<w<1$, is one-to-one.
	\end{theorem}
	\begin{proof}
	Let us suppose that 
	$F_W(w;\boldsymbol{\theta}_1)=F_W(w;\boldsymbol{\theta}_2)
	$, for all $0 < w < 1$, where $\boldsymbol{\theta}_i=(\sigma_i,\alpha_i,\rho_i)^\top\in\Theta$, $i=1,2$. From \eqref{CDF} the following identity  (for $0 < w < 1$)
	\begin{align*}
		\dfrac{\displaystyle {s^{\alpha_1}\over\sigma_1^{\alpha_1}} }{\displaystyle {s^{\alpha_1}\over \sigma_1^{\alpha_1}}+1}
	\,
	\dfrac{\displaystyle\left({s^{\alpha_1}\over \sigma_1^{\alpha_1}}+1\right)^{2}-\rho_1}
	{
		\displaystyle\left({s^{\alpha_1}\over\sigma_1^{\alpha_1}}+1\right)^{2}-{\rho_1}\,{s^{\alpha_1}\over{\sigma_1}^{\alpha_1}}
	}
	=
			\dfrac{\displaystyle {s^{\alpha_2}\over\sigma_2^{\alpha_2}} }{\displaystyle {s^{\alpha_2}\over \sigma_2^{\alpha_2}}+1}
	\,
	\dfrac{\displaystyle\left({s^{\alpha_2}\over \sigma_2^{\alpha_2}}+1\right)^{2}-\rho_2}
	{
		\displaystyle\left({s^{\alpha_2}\over\sigma_2^{\alpha_2}}+1\right)^{2}-{\rho_2}\,{s^{\alpha_2}\over{\sigma_2}^{\alpha_2}}
	},
	\quad 
	s={w\over 1-w},
	\end{align*}
	is valid. The above identity can be written as
	\begin{align}
	&
	s^{\alpha_1}
	\Bigg[
	{\rho_1\rho_2-3\rho_1-\rho_2+3\over \sigma_1^{\alpha_1}\sigma_2^{2\alpha_2}}\, 
	s^{2\alpha_2}
	+
	{\rho_1\rho_2-3\rho_1-\rho_2+3\over \sigma_1^{\alpha_1}\sigma_2^{\alpha_2}}\, 
	s^{\alpha_2}
	+
	{1-\rho_1\over \sigma_1^{\alpha_1}\sigma_2^{3\alpha_2}}\, s^{3\alpha_2}
	+
	{3-\rho_2\over \sigma_1^{3\alpha_1}\sigma_2^{2\alpha_2}}\, s^{2(\alpha_1+\alpha_2)}
		\nonumber
	\\[0,2cm]
	&+
	{3-\rho_2\over \sigma_1^{3\alpha_1}\sigma_2^{\alpha_2}}\, s^{2\alpha_1+\alpha_2}
	+
	{6-2\rho_2\over \sigma_1^{2\alpha_1}\sigma_2^{2\alpha_2}}\, s^{\alpha_1+2\alpha_2}
		+
	{6-2\rho_2\over \sigma_1^{2\alpha_1}\sigma_2^{\alpha_2}}\, s^{\alpha_1+\alpha_2}
	+
	{1\over \sigma_1^{3\alpha_1}\sigma_2^{3\alpha_2}}\, s^{2\alpha_1+3\alpha_2}
	+
	{2\over \sigma_1^{2\alpha_1}\sigma_2^{3\alpha_2}}\, s^{\alpha_1+3\alpha_2}
		\nonumber
		\\[0,2cm]
	&
	+
	{1\over \sigma_1^{3\alpha_1}}\, s^{2\alpha_1}
	+
	{2\over \sigma_1^{2\alpha_1}}\, s^{\alpha_1}
	+
	{1-\rho_1\over\sigma_1^{\alpha_1}}
	\Bigg]
		\nonumber
			\\[0,2cm]
	&=
		s^{\alpha_2}
	\Bigg[
	{\rho_2\rho_1-3\rho_2-\rho_1+3\over \sigma_2^{\alpha_2}\sigma_1^{2\alpha_1}}\, 
	s^{2\alpha_1}
	+
	{\rho_2\rho_1-3\rho_2-\rho_1+3\over \sigma_2^{\alpha_2}\sigma_1^{\alpha_1}}\, 
	s^{\alpha_1}
	+
	{1-\rho_2\over \sigma_2^{\alpha_2}\sigma_1^{3\alpha_1}}\, s^{3\alpha_1}
	+
	{3-\rho_1\over \sigma_2^{3\alpha_2}\sigma_1^{2\alpha_1}}\, s^{2(\alpha_2+\alpha_1)}
		\nonumber
	\\[0,2cm]
	&+
	{3-\rho_1\over \sigma_2^{3\alpha_2}\sigma_1^{\alpha_1}}\, s^{2\alpha_2+\alpha_1}
	+
	{6-2\rho_1\over \sigma_2^{2\alpha_2}\sigma_1^{2\alpha_1}}\, s^{\alpha_2+2\alpha_1}
	+
	{6-2\rho_1\over \sigma_2^{2\alpha_2}\sigma_1^{\alpha_1}}\, s^{\alpha_2+\alpha_1}
	+
	{1\over \sigma_2^{3\alpha_2}\sigma_1^{3\alpha_1}}\, s^{2\alpha_2+3\alpha_1}
	+
	{2\over \sigma_2^{2\alpha_2}\sigma_1^{3\alpha_1}}\, s^{\alpha_2+3\alpha_1}
	\nonumber
	\\[0,2cm]
	&
	+
	{1\over \sigma_2^{3\alpha_2}}\, s^{2\alpha_2}
	+
	{2\over \sigma_2^{2\alpha_2}}\, s^{\alpha_2}
	+
	{1-\rho_2\over\sigma_2^{\alpha_2}}
	\Bigg]. \label{id-main}
	\end{align}
	Now, suppose $\alpha_1\neq \alpha_2$ and $\boldsymbol{\theta}_i\in\Theta_*$, $i=1,2$. Then, from \eqref{id-main}, we have
		\begin{align*}
		&
		{\rho_1\rho_2-3\rho_1-\rho_2+3\over \sigma_1^{\alpha_1}\sigma_2^{2\alpha_2}}\, 
		s^{\alpha_1+\alpha_2}
		+
		{\rho_1\rho_2-3\rho_1-\rho_2+3\over \sigma_1^{\alpha_1}\sigma_2^{\alpha_2}}\, 
		s^{\alpha_1}
		+
		{1-\rho_1\over \sigma_1^{\alpha_1}\sigma_2^{3\alpha_2}}\, s^{\alpha_1+2\alpha_2}
		+
		{3-\rho_2\over \sigma_1^{3\alpha_1}\sigma_2^{2\alpha_2}}\, s^{3\alpha_1+\alpha_2}
		\\[0,2cm]
		&+
		{3-\rho_2\over \sigma_1^{3\alpha_1}\sigma_2^{\alpha_2}}\, s^{3\alpha_1}
		+
		{6-2\rho_2\over \sigma_1^{2\alpha_1}\sigma_2^{2\alpha_2}}\, s^{2\alpha_1+\alpha_2}
		+
		{6-2\rho_2\over \sigma_1^{2\alpha_1}\sigma_2^{\alpha_2}}\, s^{2\alpha_1}
		+
		{1\over \sigma_1^{3\alpha_1}\sigma_2^{3\alpha_2}}\, s^{3\alpha_1+2\alpha_2}
		+
		{2\over \sigma_1^{2\alpha_1}\sigma_2^{3\alpha_2}}\, s^{2(\alpha_1+\alpha_2)}
		\\[0,2cm]
		&
		+
		{1\over \sigma_1^{3\alpha_1}}\, 
		s^{3\alpha_1-\alpha_2}
		+
		{2\over \sigma_1^{2\alpha_1}}\, 
		s^{2\alpha_1-\alpha_2}
		+
		{1-\rho_1\over\sigma_1^{\alpha_1}}\,
		s^{\alpha_1-\alpha_2}
		\\[0,2cm]
		&=
		{\rho_2\rho_1-3\rho_2-\rho_1+3\over \sigma_2^{\alpha_2}\sigma_1^{2\alpha_1}}\, 
		s^{2\alpha_1}
		+
		{\rho_2\rho_1-3\rho_2-\rho_1+3\over \sigma_2^{\alpha_2}\sigma_1^{\alpha_1}}\, 
		s^{\alpha_1}
		+
		{1-\rho_2\over \sigma_2^{\alpha_2}\sigma_1^{3\alpha_1}}\, s^{3\alpha_1}
		+
		{3-\rho_1\over \sigma_2^{3\alpha_2}\sigma_1^{2\alpha_1}}\, s^{2(\alpha_2+\alpha_1)}
		\\[0,2cm]
		&+
		{3-\rho_1\over \sigma_2^{3\alpha_2}\sigma_1^{\alpha_1}}\, s^{2\alpha_2+\alpha_1}
		+
		{6-2\rho_1\over \sigma_2^{2\alpha_2}\sigma_1^{2\alpha_1}}\, s^{\alpha_2+2\alpha_1}
		+
		{6-2\rho_1\over \sigma_2^{2\alpha_2}\sigma_1^{\alpha_1}}\, s^{\alpha_2+\alpha_1}
		+
		{1\over \sigma_2^{3\alpha_2}\sigma_1^{3\alpha_1}}\, s^{2\alpha_2+3\alpha_1}
		+
		{2\over \sigma_2^{2\alpha_2}\sigma_1^{3\alpha_1}}\, s^{\alpha_2+3\alpha_1}
		\\[0,2cm]
		&
		+
		{1\over \sigma_2^{3\alpha_2}}\, s^{2\alpha_2}
		+
		{2\over \sigma_2^{2\alpha_2}}\, s^{\alpha_2}
		+
		{1-\rho_2\over\sigma_2^{\alpha_2}},
	\end{align*}
	from where, by equality of the constant terms of polynomials, we have $\rho_2=1$, which is a contradiction, since $\boldsymbol{\theta}_i\in\Theta_*$,  $i=1,2$. Therefore, we have $\alpha_1=\alpha_2$. Consequently, by using the notation $\alpha_1=\alpha_2\equiv\alpha$,  from \eqref{id-main} it follows that
		\begin{align*}
		&
		\left[
		{\rho_1\rho_2-3\rho_1-\rho_2+3\over \sigma_1^{\alpha}\sigma_2^{2\alpha}}
			+
		{6-2\rho_2\over \sigma_1^{2\alpha}\sigma_2^{\alpha}}
		+
		{1\over \sigma_1^{3\alpha}}
		\right]
		s^{2\alpha}
		+
		\left[
		{\rho_1\rho_2-3\rho_1-\rho_2+3\over \sigma_1^{\alpha}\sigma_2^{\alpha}}
		+
		{2\over \sigma_1^{2\alpha}}
		\right]
		s^{\alpha}
		\\[0,2cm]
		&
		+
		\left[
		{1-\rho_1\over \sigma_1^{\alpha}\sigma_2^{3\alpha}}
		+
		{3-\rho_2\over \sigma_1^{3\alpha}\sigma_2^{\alpha}}
		+
		{6-2\rho_2\over \sigma_1^{2\alpha}\sigma_2^{2\alpha}} 
		\right]
		s^{3\alpha}
		+
		{1\over \sigma_1^{3\alpha}\sigma_2^{3\alpha}}\, s^{5\alpha}
		+
		\left[
		{3-\rho_2\over \sigma_1^{3\alpha}\sigma_2^{2\alpha}}
		+
		{2\over \sigma_1^{2\alpha}\sigma_2^{3\alpha}}
		\right]
		 s^{4\alpha}
		+
		{1-\rho_1\over\sigma_1^{\alpha}}
		\nonumber
		\\[0,2cm]
		&=
		\left[
		{\rho_2\rho_1-3\rho_2-\rho_1+3\over \sigma_2^{\alpha}\sigma_1^{2\alpha}}
				+
		{6-2\rho_1\over \sigma_2^{2\alpha}\sigma_1^{\alpha}}
			+
		{1\over \sigma_2^{3\alpha}}
		\right] 
		s^{2\alpha}
		+
		\left[
		{\rho_2\rho_1-3\rho_2-\rho_1+3\over \sigma_2^{\alpha}\sigma_1^{\alpha}}
		+
		{2\over \sigma_2^{2\alpha}}
		\right]		
		s^{\alpha}
		\\[0,2cm]
		&
		+
		\left[
		{1-\rho_2\over \sigma_2^{\alpha}\sigma_1^{3\alpha}}
		+
		{3-\rho_1\over \sigma_2^{3\alpha}\sigma_1^{\alpha}}
		+
		{6-2\rho_1\over \sigma_2^{2\alpha}\sigma_1^{2\alpha}}
		\right]
		s^{3\alpha}
		+
		{1\over \sigma_2^{3\alpha}\sigma_1^{3\alpha}}\,
		s^{5\alpha}
		+
		\left[
		{3-\rho_1\over \sigma_2^{3\alpha}\sigma_1^{2\alpha}}
		+
		{2\over \sigma_2^{2\alpha}\sigma_1^{3\alpha}}
		\right]
		s^{4\alpha}
		+
		{1-\rho_2\over\sigma_2^{\alpha}}.
	\end{align*}
	By the equality of the constant terms of polynomials, we have
	\begin{align}\label{eq1}
		{1-\rho_1\over\sigma_1^{\alpha}}
		=
	{1-\rho_2\over\sigma_2^{\alpha}}
	\quad \Longleftrightarrow \quad
	{1\over\sigma_1^\alpha}-{1\over\sigma_2^\alpha}
	=
	{\rho_1\over\sigma_1^\alpha}-{\rho_2\over\sigma_2^\alpha}.
	\end{align}
	Equating the coefficients of the monomial $s^{4\alpha}$, it follows that
	\begin{align}\label{eq2}
	{1\over\sigma_1^\alpha}-{1\over\sigma_2^\alpha}
	=
	{\rho_2\over\sigma_1^\alpha}-{\rho_1\over\sigma_2^\alpha}.
	\end{align}
	By combining \eqref{eq1} and \eqref{eq2}, we obtain
	\begin{align*}
	{\rho_2-\rho_1\over \sigma_1^\alpha}
	=
	-\left({\rho_2-\rho_1\over \sigma_2^\alpha}\right).
	\end{align*}
	If $\rho_1\neq \rho_2$ then $\sigma_2^\alpha=-\sigma_1^\alpha$,  which is a contradiction since $\sigma_1$ and $\sigma_2$ are positive. Therefore, $\rho_1=\rho_2$. Substituting this identity into \eqref{eq1}, we have $\sigma_1^{\alpha}=\sigma_2^{\alpha}$, from which we obtain $\sigma_1=\sigma_2$. So we have completed the proof.
	\end{proof}

	\subsection{Symmetry}
	
	The following result shows the symmetry of the UF distribution when $\sigma= 1$. For $\sigma\neq 1$, the ULS distribution presents asymmetry.
	\begin{proposition}
	The UF PDF in \eqref{pdf-main} is symmetric around $w= 1/2$ provided
	$\sigma= 1$. 
	\end{proposition}
	\begin{proof}
		Set
		\begin{align*}
		s_0={{1\over 2}-\delta\over {1\over 2}+\delta}, \quad \forall 0<\delta<{1\over 2}.
		\end{align*}
	By using \eqref{pdf-main} with $\sigma= 1$, note that
		\begin{align*}
	f_W\left({1\over 2}-\delta;\boldsymbol{\theta}\right)
	=
	{\alpha}
	s_0^{\alpha-1}(s_0+1)^2
	\left\{ \displaystyle
	{\displaystyle 2\left({s_0^\alpha}+1\right)^2
		-
		\rho (s_0^{2\alpha}+1)
		\over 
		\displaystyle
		\left[
		\left({s_0^\alpha}+1\right)^2
		-\rho {s_0^\alpha}
		\right]^2}
	- 
	{\displaystyle 1\over\displaystyle \left({s_0^\alpha} + 1\right)^2}
	\right\}
	\end{align*}
	and
	\begin{align*}
		f_W\left({1\over 2}+\delta;\boldsymbol{\theta}\right)
	&=
	{\alpha}\,
	{1\over s_0^{\alpha-1}}\left({1\over s_0}+1\right)^2
	\left\{ \displaystyle
	{\displaystyle 2\left({1\over s_0^\alpha}+1\right)^2
		-
		\rho \left({1\over s_0^{2\alpha}}+1\right)
		\over 
		\displaystyle
		\left[
		\left({1\over s_0^\alpha}+1\right)^2
		-\rho\, {1\over s_0^\alpha}
		\right]^2}
	- 
	{\displaystyle 1\over\displaystyle \left({1\over s_0^\alpha} + 1\right)^2}
	\right\}.
	\end{align*}
	From the last two identities above, simple algebraic manipulations show that
	\begin{align*}
	f_W\left({1\over 2}-\delta;\boldsymbol{\theta}\right)
	=		
	f_W\left({1\over 2}+\delta;\boldsymbol{\theta}\right),
	\quad \forall 0<\delta<{1\over 2}.
	\end{align*}
	The required result then follows readily.	
	\end{proof}

	\subsection{Stochastic representation}
	
	Let $Y$ be the corresponding random variable of the CDF $G(\cdot;a)$ defined in \eqref{def-G-cdf}.
	A simple observation shows that if $W\sim {\rm UF}(\boldsymbol{\theta})$, then, from \eqref{CDF}, the CDF of $W$ can be expressed as (for $0 < w < 1$)
	\begin{align*}
		F_W(w;\boldsymbol{\theta})
		=
		\mathbb{P}\left(Y\leqslant{s^\alpha\over \sigma^\alpha}\right)
		=
		\mathbb{P}\left({\sigma Y^{1/\alpha}\over 1+ \sigma Y^{1/\alpha}}\leqslant w\right).
	\end{align*}
	In other words,
	\begin{align}\label{stochastic-rep}
	W\stackrel{d}{=}{\sigma Y^{1/\alpha}\over 1+ \sigma Y^{1/\alpha}},
	\end{align}
	where $\stackrel{d}{=}$
	means equality in distribution. Analogously, we also
	have
	\begin{align*}
	Y\stackrel{d}{=}\left[{W\over \sigma(1-W)}\right]^\alpha.
	\end{align*}

	\subsection{UF model arising as a ratio}
	In this subsection, we establish one of the most important properties of the proposed model. This property characterizes a UF random variable $W$ in \eqref{CDF} as a ratio of the
	form $X_1/(X_1 + X_2)$, where the random vector 
	$(X_1, X_2)^\top$ follows a bivariate extreme distribution (with Fréchet margins) as discussed in \cite[][p. 14, Item 21]{Elshamy1992}.
	
    Following \cite{Elshamy1992}, a random vector $(X_1, X_2)^\top$ has a  bivariate extreme distribution if its CDF, denoted by $F_{X_1,X_2}$, is given by \eqref{id-00}, that is,  
	\begin{align}\label{id-0}
		F_{X_1,X_2}(x_1,x_2)
		=
		\exp\left\{-\left({x_1\over\sigma_1}\right)^{-\alpha}-\left({x_2\over\sigma_2}\right)^{-\alpha} +
		\rho\left[\left({x_1\over\sigma_1}\right)^\alpha+\left({x_2\over\sigma_2}\right)^\alpha\right]^{-1}\right\},
	\end{align}
	where $x_1\geqslant0$, $x_2\geqslant 0$;
	$\sigma_1>0,\sigma_2>0, \alpha>0$ and $0\leqslant \rho\leqslant 1$.
	
	Setting
	\begin{align*}
		T\equiv{X_1\over X_1+X_2},
	\end{align*}
	we have (for $0<w<1$)
	\begin{align}\label{id-1}
		F_T(w)
		=
		\mathbb{P}\left({X_2\over X_1}\leqslant s\right)
		=
		\int_{0}^{\infty}
		\left[
		\int_{0}^{sx_2}
		f_{X_1,X_2}(x_1,x_2){\rm d}x_1
		\right] 
		{\rm d}x_2,
		\quad
		s={w\over 1-w}.
	\end{align}
	
	Note that (for any $r>0$)
	\begin{align}\label{id-2}
		\int_{0}^{r}
		f_{X_1,X_2}(x_1,x_2){\rm d}x_1
		=
		{\partial\over\partial x_2}
		\int_{0}^{r}
		{\partial\over\partial x_1}
		F_{X_1,X_2}(x_1,x_2){\rm d}x_1
		=
		{\partial\over\partial x_2}
		F_{X_1,X_2}(r,x_2).
	\end{align}
	Then, by using \eqref{id-2} in \eqref{id-1}, we have
	\begin{align*}
		F_T(w)
		=
		\int_{0}^{\infty}
		{\partial\over\partial x_2}
		F_{X_1,X_2}(r,x_2) 
		\bigg\vert_{r=sx_2}
		{\rm d}x_2.
	\end{align*}
		Using \eqref{id-0}, the integral above becomes
		\begin{align}
		&=
		\int_{0}^{\infty}
		F_{X_1,X_2}(r,x_2)  \,
		{\alpha\over  x_2}\left[\left(x_2\over\sigma_2\right)^{-\alpha}-{\rho\left(x_2\over\sigma_2\right)^{\alpha}\left[\left(r\over\sigma_1\right)^\alpha+\left(x_2\over\sigma_2\right)^\alpha\right]^{-2}}\right]
		\Bigg\vert_{r=sx_2}
		{\rm d}x_2
		\nonumber
		\\[0,2cm]
		&=
		{\displaystyle
			1-{{\rho\over\sigma_2^{2\alpha}}\left[\left(s\over\sigma_1\right)^\alpha+{1\over\sigma_2^\alpha} \right]^{-2}}
			\over 
			\displaystyle
			1+
			\left(
			{s\sigma_2\over\sigma_1}\right)^{-\alpha}
			-
			{\rho\over\sigma_2^\alpha}\left[\left({s\over\sigma_1}\right)^\alpha+{1\over\sigma_2^\alpha}\right]^{-1}
		}. \label{id-3}
	\end{align}
	
	From \eqref{id-1} and \eqref{id-3}, we have proven that
	\begin{align}\label{id-4}
		F_T(w)
		=
		{\displaystyle
			1-{\rho\left[\left({\sigma_2\over\sigma_1}\, s\right)^\alpha+1 \right]^{-2}}
			\over 
			\displaystyle
			1+
			\left(
			{ \sigma_2\over\sigma_1}\, s\right)^{-\alpha}
			-
			\rho\left[\left({\sigma_2\over\sigma_1}\, s\right)^\alpha+1\right]^{-1}
		}.
	\end{align}
	Setting $\sigma\equiv\sigma_1/\sigma_2$, the above identity is written as (for $0 < w < 1$)
	\begin{align}\label{id-5}
		F_T(w)
		=
		\dfrac{\displaystyle {s^\alpha\over\sigma^\alpha} }{\displaystyle\left({s^\alpha\over\sigma^\alpha}+1\right)}
		\,
		\dfrac{\displaystyle\left({s^\alpha\over\sigma^\alpha}+1\right)^{2}-\rho}
		{
			\displaystyle\left({s^\alpha\over\sigma^\alpha}+1\right)^{2}-\rho\,{s^\alpha\over\sigma^\alpha}
		}
		=
		F_W(w;\boldsymbol{\theta}),
			\quad
		s={w\over 1-w},
	\end{align}
where in the last equality we have used the formula in \eqref{CDF}.
That is,
\begin{align}\label{rep-est}
	W\stackrel{d}{=}{X_1\over X_1+X_2}.
\end{align}
In other words, distributionally, the UF random variable $W$ can be treated as the ratio in \eqref{rep-est} involving two random variables $X_1$ and $X_2$ with $(X_1, X_2)^\top$ following a bivariate extreme distribution \citep{Elshamy1992}.

\begin{remark}
	From \eqref{stochastic-rep} and \eqref{rep-est}, respectively, we have
\begin{align*}
	W\stackrel{d}{=}{1\over 1+{1\over \sigma Y^{1/\alpha}}}
	\quad \text{and} \quad 
	W\stackrel{d}{=}{1\over 1+{X_2\over X_1}}.
\end{align*}
From which it follows that
\begin{align*}
Y \stackrel{d}{=} \left({X_1\over \sigma X_2}\right)^\alpha.
\end{align*}
\end{remark}

\subsection{Approximations for the moments}\label{Approximations for the moments}

Let $g:[0,\infty)^2\to\mathbb{R}$ be a twice differentiable integrable function. The second order Taylor expansion of $g(\boldsymbol{x})$, with $\boldsymbol{x}=(x_1,x_2)^\top$, around $\boldsymbol{\mu}=(\mu_1,\mu_2)^\top=(\mathbb{E}(X_1),\mathbb{E}(X_2))^\top$ is
\begin{align*}
    g(\boldsymbol{x})
    &=
    g(\boldsymbol{\mu})
    +
    {\partial g(\boldsymbol{x})\over\partial x_1}\Bigg\vert_{\boldsymbol{x}=\boldsymbol{\mu}}
    (x_1-\mu_1)
    +
    {\partial g(\boldsymbol{x})\over\partial x_2}\Bigg\vert_{\boldsymbol{x}=\boldsymbol{\mu}}
    (x_2-\mu_2)
    \\[0,2cm]
    &+
    {1\over 2}\left\{
    {\partial^2 g(\boldsymbol{x})\over\partial x_1^2}\Bigg\vert_{\boldsymbol{x}=\boldsymbol{\mu}}
    (x_1-\mu_1)^2
    +
    2\,
    {\partial^2 g(\boldsymbol{x})\over\partial x_1\partial x_2}\Bigg\vert_{\boldsymbol{x}=\boldsymbol{\mu}}(x_1-\mu_1)(x_2-\mu_2)
    +
    {\partial^2 g(\boldsymbol{x})\over\partial x_2^2}\Bigg\vert_{\boldsymbol{x}=\boldsymbol{\mu}}
    (x_2-\mu_2)^2
    \right\}
    \\[0,2cm]
    &+
    R_2(\boldsymbol{x}),
\end{align*}
where $R_2(\boldsymbol{x})$ is a remainder of a smaller order than the terms in the equation. Hence, a better approximation for $\mathbb{E}[g(\boldsymbol{X})]$, with $\boldsymbol{X}=(X_1,X_2)^\top$ being as in \eqref{rep-est},  is 
\begin{align}\label{aprox-formula}
    \mathbb{E}[g(\boldsymbol{X})]
        &\approx
    g(\boldsymbol{\mu})
    +
    {1\over 2}\left\{
    {\partial^2 g(\boldsymbol{x})\over\partial x_1^2}\Bigg\vert_{\boldsymbol{x}=\boldsymbol{\mu}}
    {\rm Var}(X_1)
    +
    2\,
    {\partial^2 g(\boldsymbol{x})\over\partial x_1\partial x_2}\Bigg\vert_{\boldsymbol{x}=\boldsymbol{\mu}}
    {\rm Cov}(X_1,X_2)
    +
    {\partial^2 g(\boldsymbol{x})\over\partial x_2^2}\Bigg\vert_{\boldsymbol{x}=\boldsymbol{\mu}}
    {\rm Var}(X_2)
    \right\}.
    \end{align}

    By taking  $g(\boldsymbol{x})=[x_1/(x_1+x_2)]^p$, $p\in\mathbb{R}$, in \eqref{aprox-formula}, from \eqref{rep-est} the real moments of $W\sim{\rm UF}(\boldsymbol{\theta})$ can be approximated as follows
    \begin{align*}
    &\mathbb{E}(W^p)
        =
    \mathbb{E}\left(\left[{X_1\over X_1+X_2}\right]^p\right)
    \approx
    \left({\mu_1\over\mu_1+\mu_2}\right)^p
    \\[0,2cm]
    &+
    {p\mu_1^{p - 1}\over 2(\mu_1  + \mu_2 )^{p+2}}
    \left\{
{\mu_2\over \mu_1}\, [(p - 1)  \mu_2 - 2 \mu_1] 
    {\rm Var}(X_1)
    +
    2
    (\mu_1  - p \mu_2)
    {\rm Cov}(X_1,X_2)
    +
{(p + 1) \mu_1} 
    {\rm Var}(X_2)
    \right\}.
 \end{align*}

In particular,
    \begin{align*}
    \mathbb{E}(W)
        &\approx
    {\mu_1\over\mu_1+\mu_2}
    -
    {1\over (\mu_1 + \mu_2)^3}
    \left\{
\mu_2
    {\rm Var}(X_1)
    +
    {(\mu_2-\mu_1)} \, 
    {\rm Cov}(X_1,X_2)
    -
\mu_1 
    {\rm Var}(X_2)
    \right\}
    \end{align*}
    and
        \begin{align*}
    \mathbb{E}(W^2)
        &\approx
    \left({\mu_1\over\mu_1+\mu_2}\right)^2
    \\[0,2cm]
    &+
    {\mu_1 \over (\mu_1  + \mu_2 )^4}
    \left\{
{\mu_2\over \mu_1} \, (\mu_2 - 2 \mu_1)
    {\rm Var}(X_1)
    +
    2
    (\mu_1  - 2 \mu_2)
    {\rm Cov}(X_1,X_2)
    +
{3\mu_1} 
    {\rm Var}(X_2)
    \right\}.
 \end{align*}

 Therefore,
 \begin{align*}
     {\rm Var}(W)
        &\approx
    {\mu_1 \over (\mu_1  + \mu_2 )^4}
    \left\{
{\mu_2\over \mu_1} \, (\mu_2 - 2 \mu_1)
    {\rm Var}(X_1)
    +
    2
    (\mu_1  - 2 \mu_2)
    {\rm Cov}(X_1,X_2)
    +
{3\mu_1} 
    {\rm Var}(X_2)
    \right\} 
    \\[0,2cm]
    &  -
    {1\over (\mu_1 + \mu_2)^6}
    \left\{
\mu_2
    {\rm Var}(X_1)
    +
    {(\mu_2-\mu_1)} \, 
    {\rm Cov}(X_1,X_2)
    -
\mu_1 
    {\rm Var}(X_2)
    \right\}^2
        \\[0,2cm]
    & 
    +
    {2\mu_1\over (\mu_1 + \mu_2)^4}
    \left\{
\mu_2
    {\rm Var}(X_1)
    +
    {(\mu_2-\mu_1)} \, 
    {\rm Cov}(X_1,X_2)
    -
\mu_1 
    {\rm Var}(X_2)
    \right\}.    
    \end{align*}
        Since $\boldsymbol{X}=(X_1, X_2)^\top$ is as in \eqref{rep-est}, that is, $\boldsymbol{X}$  follows a bivariate extreme distribution with Fréchet margins \citep[see ][p. 14, Item 21]{Elshamy1992}, in the above identities we are considering that (for $i=1,2$)
        \begin{align*}
        &\mu_i=\sigma_i \Gamma\left(1-{1\over\alpha}\right), \quad \alpha>1,
        \\[0,2cm]
        &{\rm Var}(X_i)=\sigma_i^2\left[\Gamma\left(1-{2\over\alpha}\right)-\Gamma^2\left(1-{1\over\alpha}\right)\right], \quad \alpha>2,
        \\[0,2cm]
        &\vert{\rm Cov}(X_1,X_2)\vert\leqslant
        {\sigma_1\sigma_2\left[\Gamma\left(1-{2\over\alpha}\right)-\Gamma^2\left(1-{1\over\alpha}\right)\right]}, \quad \alpha>2.
        \end{align*}

\subsection{The stress-strength probability}

In what follows we explicitly obtain the stress-strength probability $R = \mathbb{P}(X_1 < X_2 )$ in the case where $(X_1, X_2)^\top$ follows a bivariate extreme distribution \citep{Elshamy1992}.

Indeed, by \eqref{id-1} and \eqref{id-5}, we have
\begin{align*}
    R=\mathbb{P}\left({X_2\over X_1}>1\right)
    =
    1-\mathbb{P}\left({X_2\over X_1}\leqslant 1\right)
    =
    1-F_T(w)\Big\vert_{w=1/2}
    =1-F_W(w;\boldsymbol{\theta})\Big\vert_{w=1/2}.
\end{align*}
Hence, as $w = 1/2$ if and only if $s = 1$, by using \eqref{CDF}, we get
\begin{align*}
    R=		
	\dfrac{\displaystyle 
		\displaystyle\left({1\over\sigma^\alpha}+1\right)^{3}
    -
    {1\over\sigma^\alpha}\left({1\over\sigma^\alpha}+1\right)^{2}
            -\rho\,{1\over\sigma^{2\alpha}}
    }
	{
	\displaystyle
\left({1\over\sigma^\alpha}+1\right)
    \left[
    \left({1\over\sigma^\alpha}+1\right)^{2}-\rho\,{1\over\sigma^\alpha} 
    \right]
	}.
\end{align*}
In the special case $\rho=0$,
\begin{align*}
    R=	
    {\sigma^\alpha\over \sigma^\alpha+1}.
\end{align*}



	\subsection{Quantiles}
	The quantile function of a distribution is useful for the generation of random numbers from the distribution. Since quantiles are invariant with respect to monotonic transformations, from \eqref{stochastic-rep}, the quantile of $W\sim {\rm UF}(\boldsymbol{\theta})$, denoted by $Q_W(p)$, $0 < p < 1$, is written as
	\begin{align}\label{stochastic-rep-main}
	Q_W(p)
	=
	{\sigma Q_Y^{1/\alpha}(p)\over 1+ \sigma Q_Y^{1/\alpha}(p)},
	\end{align}
	where $Y$ has CDF $G(\cdot;\rho)$ given in \eqref{def-G-cdf}.
Note that $Q_Y(p)$ can be obtained by inverting the equation $G(Q_Y(p);\rho)=p$. Consequently, $Q_Y(p)$ is the only zero (in the positive half interval) of the following cubic polynomial:	
\begin{align*}
	(p-1)x^3+[(3-\rho)p-2]x^2+[(3-\rho)p+\rho-1]x+p, \quad x>0.
\end{align*}

\subsection{Maximum likelihood estimation}\label{mle}

Let $\{W_i: i = 1,\ldots, n\}$ be a random sample of size $n$ from $W\sim {\rm UF}(\boldsymbol{\theta})$ and let $w_i$ be the sample observation of $W_i$, $i=1,\ldots,n$. From \eqref{f-expression},
	\begin{align*}
	f_{W_i}(w_i;\boldsymbol{\theta})
	=
	{\alpha\over \sigma^\alpha}\, s^{\alpha-1}_i(s_i+1)^2\,
	g\left({s^\alpha_i\over \sigma^\alpha};\rho\right),
	\quad 
	s_i={w_i\over 1-w_i},
	\quad i=1,\ldots,n,
\end{align*}
where $g(\cdot;\rho)$ is given in \eqref{def-g-cdf}.
Then, the log-likelihood
function for $\boldsymbol{\theta} = (\sigma,\alpha,\rho)^\top$ is given by
\begin{align*}
	\ell(\boldsymbol{\theta})
	=
	n\log(\alpha)-n\alpha\log(\sigma)
	+
	(\alpha-1)
	\sum_{i=1}^{n}
	\log(s_i)
	+
	2
	\sum_{i=1}^{n}
	\log(s_i+1)
	+
	\sum_{i=1}^{n}
	\log\left(g\left({s^\alpha_i\over \sigma^\alpha};\rho\right)\right).
\end{align*}
The likelihood equations are
\begin{align*}
{\partial \ell(\boldsymbol{\theta})\over\partial \sigma}
&=
{n\alpha\over\sigma}
+
{\alpha\over \sigma}
\sum_{i=1}^{n}
{g'(x_i;\rho)\over g(x_i;\rho)}\,
x_i\, \bigg\vert_{x_i={s^\alpha_i\over \sigma^\alpha}},
\\[0,2cm]
{\partial \ell(\boldsymbol{\theta})\over\partial \alpha}
&=
{n\over\alpha}
+
\sum_{i=1}^{n}
\log(s_i)
+
{1\over \alpha}
\sum_{i=1}^{n}
{g'(x_i;\rho)\over g(x_i;\rho)}\,
x_i \log(\sigma^\alpha x_i) \, \bigg\vert_{x_i={s^\alpha_i\over \sigma^\alpha}}
- 
{\log(\sigma)}
\sum_{i=1}^{n}
{g'(x_i;\rho)\over g(x_i;\rho)}\,
x_i\, \bigg\vert_{x_i={s^\alpha_i\over \sigma^\alpha}},
\\[0,2cm]
{\partial \ell(\boldsymbol{\theta})\over\partial \rho}
&=
-
\sum_{i=1}^{n}
{x^4_i+ (\rho-2) x^3_i - 6 x^2_i+ (\rho-2) x_i  +1\over g(x_i;\rho) [(x_i+1)^2-\rho x_i]^3} \, \bigg\vert_{x_i={s^\alpha_i\over \sigma^\alpha}},
\end{align*}
where
\begin{align*}
	g'(x;\rho)
	=
-\dfrac{2[\rho^3 x^3 - \rho(x-2)^2 (x+1)^4 + (x+1)^6 + \rho^2 (1 + 3 x - 5 x^3 - 3 x^4)]}
{(x+1)^3 [(x+1)^2-\rho x]^3}.
\end{align*}
Necessary conditions for the occurrence
of a maximum (or a minimum) are
\begin{align*}
 	{\partial \ell(\boldsymbol{\theta})\over\partial \sigma}=0, \quad 
 	{\partial \ell(\boldsymbol{\theta})\over\partial \alpha}=0, \quad 
 	{\partial \ell(\boldsymbol{\theta})\over\partial \rho}=0.
\end{align*}
When a nontrivial root $\widehat{\boldsymbol{\theta}} = (\widehat{\sigma},\widehat{\alpha},\widehat{\rho})^\top$ of the likelihood equations provides the absolute maximum of the log-likelihood function, it is called an ML estimate in the strict sense.
Any nontrivial root of these equations is called the ML estimate of $\boldsymbol{\theta}$ in the loose
sense.  It is not possible to solve the above likelihood equations analytically, therefore we must resort to numerical optimization
methods.


\section{Simulation Results}\label{sec_simulations}

In this section, we provide a study involving Monte Carlo simulations that analyze the performance of the ML estimator $\hat{\bm \theta}$. 

To evaluate the performance of the ML estimator $\hat{\bm \theta}$, we fix several values of the parameters $\bm\theta=(\sigma$,$\alpha$,$\rho)$, and then we generate $N=1,000$ Monte-Carlo samples with sample size $n\in\{30,50,100\}$, of the random variable $W\sim UF(\sigma, \alpha, \rho)$.
We analyze the estimates $\hat{\bm \theta}$ through relative bias (RB), and root mean squared error (RMSE).

Random samples from the UF distribution can be generated using the inversion method, based on the quantiles \( F_W^{-1}(U) \), where \( U \) is a uniform random variable in \([0, 1]\). The inverse function was computed numerically using the {\it `GoFKernel'} package in the R software.

To obtain the ML estimates for the UF distribution, we used the \textit{goodness.fit} function from the \textit{AdequacyModel} package in the R software.

For the simulation, the following procedure was carried out:
 \begin{enumerate}[(1)]
 \item fix the parameter $\bm\theta$, the sample size $n\in\{30, 50, 100\}$ and the number of Monte Carlo replications $N=1,000$;
 
     \item the estimate $\hat{\boldsymbol{\theta}}^{(j)}$ is computed, for each Monte-Carlo sample:
     $$W_1^{(j)}, \cdots, W_n^{(j)}, ~~j=1,\cdots, N;$$


     \item the RB and RMSE are computed as follows:
     $$RB(\hat{\theta}_k) = \frac{1}{N}\sum_{j=1}^N \left(\frac{\hat{\theta}^{(j)}_k - \theta_k}{\theta_k}\right), ~~k=1,2,3,~\mbox{and}~ \theta_1=\sigma, \theta_2=\alpha, \theta_3=\rho, $$
     and 
     $$RMSE(\hat{\theta}_k) = \frac{1}{N}\sum_{j=1}^N \left(\hat{\theta}^{(j)}_k - \theta_k\right)^2.$$
 \end{enumerate}

The simulation results are presented in Figures \ref{fig:sim_RB} and \ref{fig:sim_RMSE}.
The ML estimator shows good behavior with minimal relative bias and low root mean squared error. Furthermore, it is clear that increasing the sample size $n$ leads RB and RMSE to approach zero.
This demonstrates the precision of the ML estimator for the UF distribution, even with small sample sizes. Additionally, we observe that the RMSE of $\rho$ estimates is less affected by variations in sample size.

\begin{figure}[htb!]
	\centering
	\includegraphics[width=1.0\linewidth]{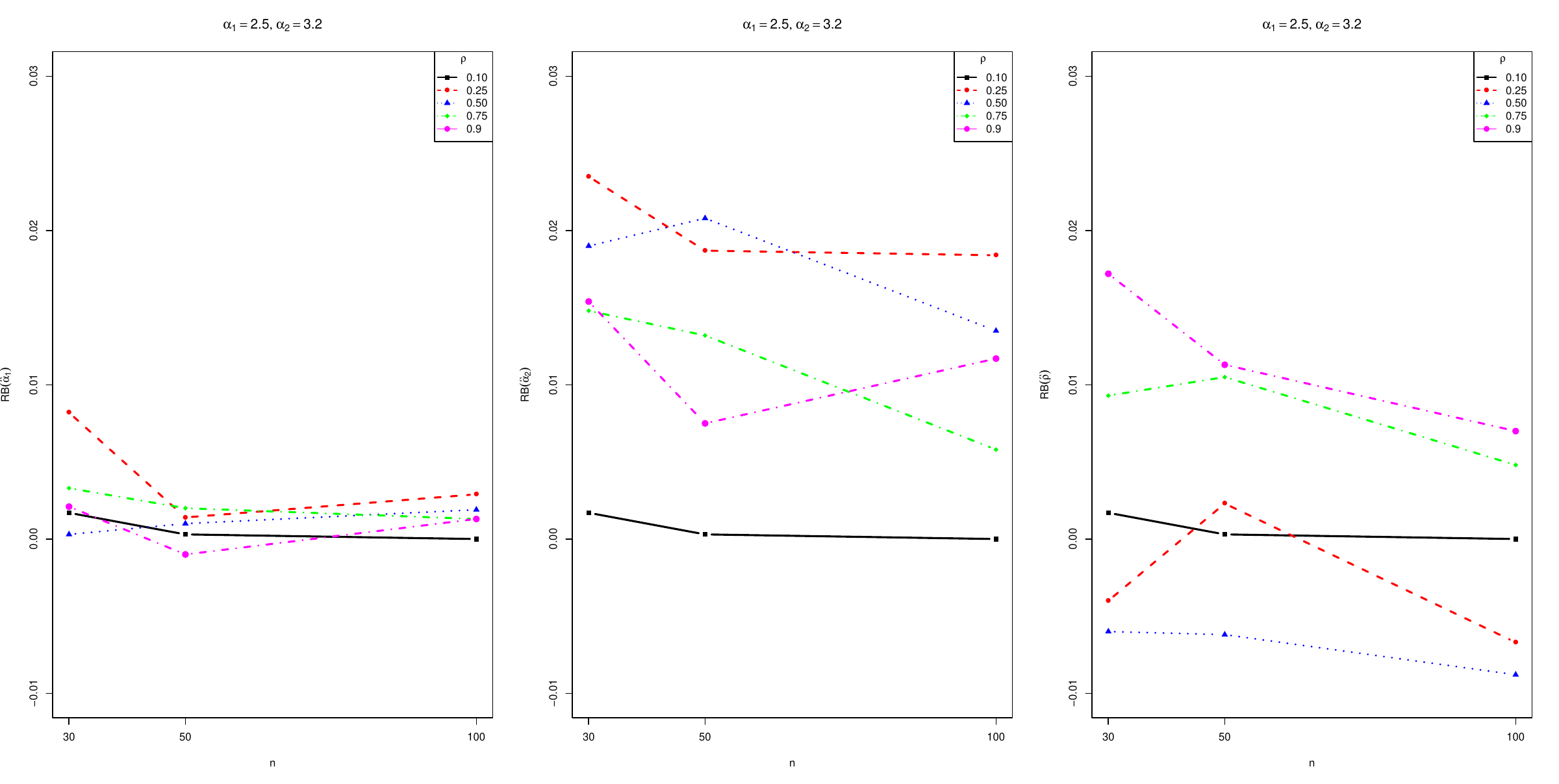}
	\caption{Relative bias for the ML estimates of $\alpha_1$ (left), $\alpha_2$ (middle) and $\rho$ (right).}
	\label{fig:sim_RB}
\end{figure}

\begin{figure}[htb!]
	\centering
	\includegraphics[width=1.0\linewidth]{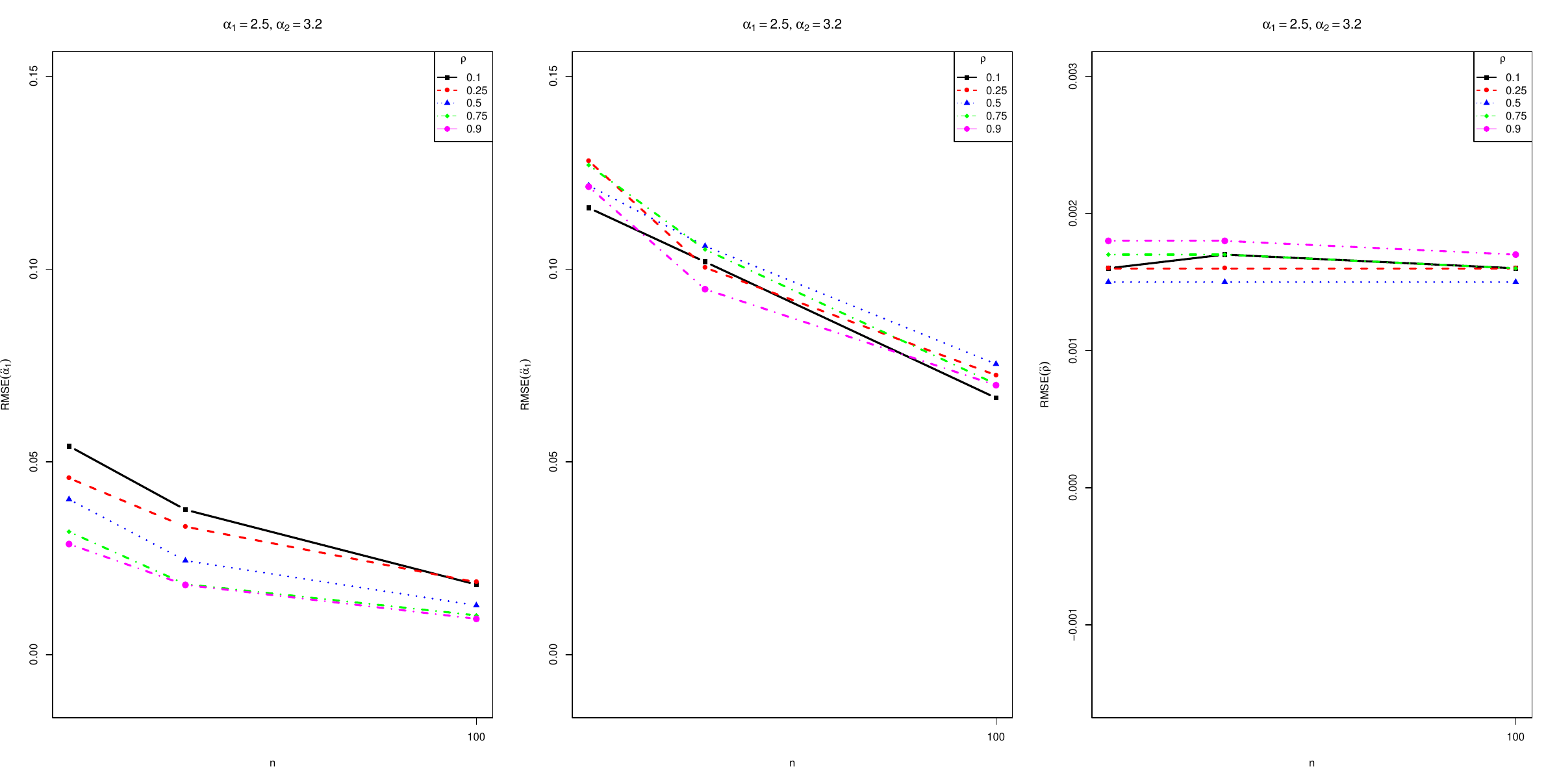}
	\caption{Root mean squared error for the ML estimates of $\alpha_1$ (left), $\alpha_2$ (middle) and $\rho$ (right).}
	\label{fig:sim_RMSE}
\end{figure}

\section{Applications}\label{sec_applications}
We now give two applications using real data previously analyzed in the literature. We study the validity of the unit-Fréchet distribution for the data sets and show that the new model presents good fits in both cases.

\subsection{Medium pass completion proportion}

Statistical tools are essential for modeling sports data sets. The better the model fits, the better the decisions can be made. 
In this context, the unit-Fréchet distribution is well-suited for modeling football proportion data. 
For this, we used the data which was modeled before in \cite{Vila24bivariate} and \cite{Quintino24estimation}.
The data\footnote{Available at \url{https://www.kaggle.com/} (accessed on 13 February 2024)} consists of the medium pass completion proportion in UEFA Champions League for the seasons 2004/05 and 2005/06, that is, the relative frequency of successful passes between 14 and 18m - thus a number from 0 (none of the passes) to 1 (all the passes).
For the convenience of the reader, the data are presented below:
\begin{eqnarray}
    \nonumber
    &&(0.289, 0.700, 0.211, 0.733,0.444, 0.544, 0.089, 0.767, 0.433, 0.911,0.800, 0.733, 0.278, 0.456, 0.178,\\
   \nonumber &&
     0.200, 0.244, 0.467, 0.022, 0.400, 0.378,  0.589, 0.600, 0.567, 
 0.844, 0.711, 0.289, 0.178, 0.489, 0.278,\\
 \nonumber  &&    0.611, 0.544,0.267, 0.489, 0.467, 0.300, 0.311)
\end{eqnarray}

Descriptive statistics are presented in Table \ref{tab:fut_descritive}. As the data set has unit support, the unit-Fréchet is a natural candidate to model such data. 
We compared the fitted model with the well-known Beta and Kumaraswamy distributions. Additionally, we evaluated the results against the p-max stable law $H_2$ distribution, as previously utilized in \cite{Quintino24estimation}. 
To obtain the estimation results for Beta and Kumaraswamy models in Table \ref{tab:fut_estimation}, we used the package \textit{AdequacyModel} of the R software and we used the function `\textit{goodness.fit}' to obtain maximum likelihood estimates.
After estimating the parameters for the four distributions, the information criteria AIC (Akaike information criterion) and BIC (Bayesian information criterion) were applied, justifying the choice of the unit-Fréchet (see Table \ref{tab:fut_estimation}).
This choice was also supported by the Kolmogorov-Smirnov test, whose p-value was $0.8837$, indicating that we could not reject the null hypothesis that the true model is the unit-Fréchet.
Figure \ref{fig:hist_ecdf_fut} shows the fit of distributions to the data set.

\begin{table}[!ht]
	\centering
	\caption{Descriptive statistics for income-consumption data.}
	\label{tab:fut_descritive}
		\begin{tabular}{cccccccccc}
			\hline
			n & Min. & 1st Qu. & Median & Mean & 3rd Qu. & Max. & Std. dv. & CS & CK\\ 
			\hline
			37& 0.02 & 0.28 & 0.46 & 0.45 & 0.60 & 0.91 & 0.22 & 0.16 & -0.92  \\ 
			\hline
	\end{tabular}
\end{table}

	\begin{table}[!ht]
\caption{Estimated parameters and information criteria for model selection.}
\label{tab:fut_estimation}
\centering
\begin{tabular}{ccccc}
  \hline
Model & $\hat{\bm\theta}$ & $\ell(\hat{\boldsymbol{\theta}})$ & AIC & BIC  \\ 
  \hline
\bf UF & (0.8064, 1.0590, 0.8235)  & 4.9208 & \bf -3.8417 & \bf 0.9911 \\ 
Beta & (1.5994, 1.8739) & 4.7113 & -3.4226 & 1.4102\\
Kumaraswamy & (1.5721, 1.9757) &  4.7834 & -3.5668 & 1.2660\\ 
$H_2$ & (1.8231, 0.7822, 0.8249) & 1.1325&  3.7351 &  8.5678\\
   \hline
\end{tabular}
\end{table}

\begin{figure}[H]

    \includegraphics[width=0.9\linewidth]{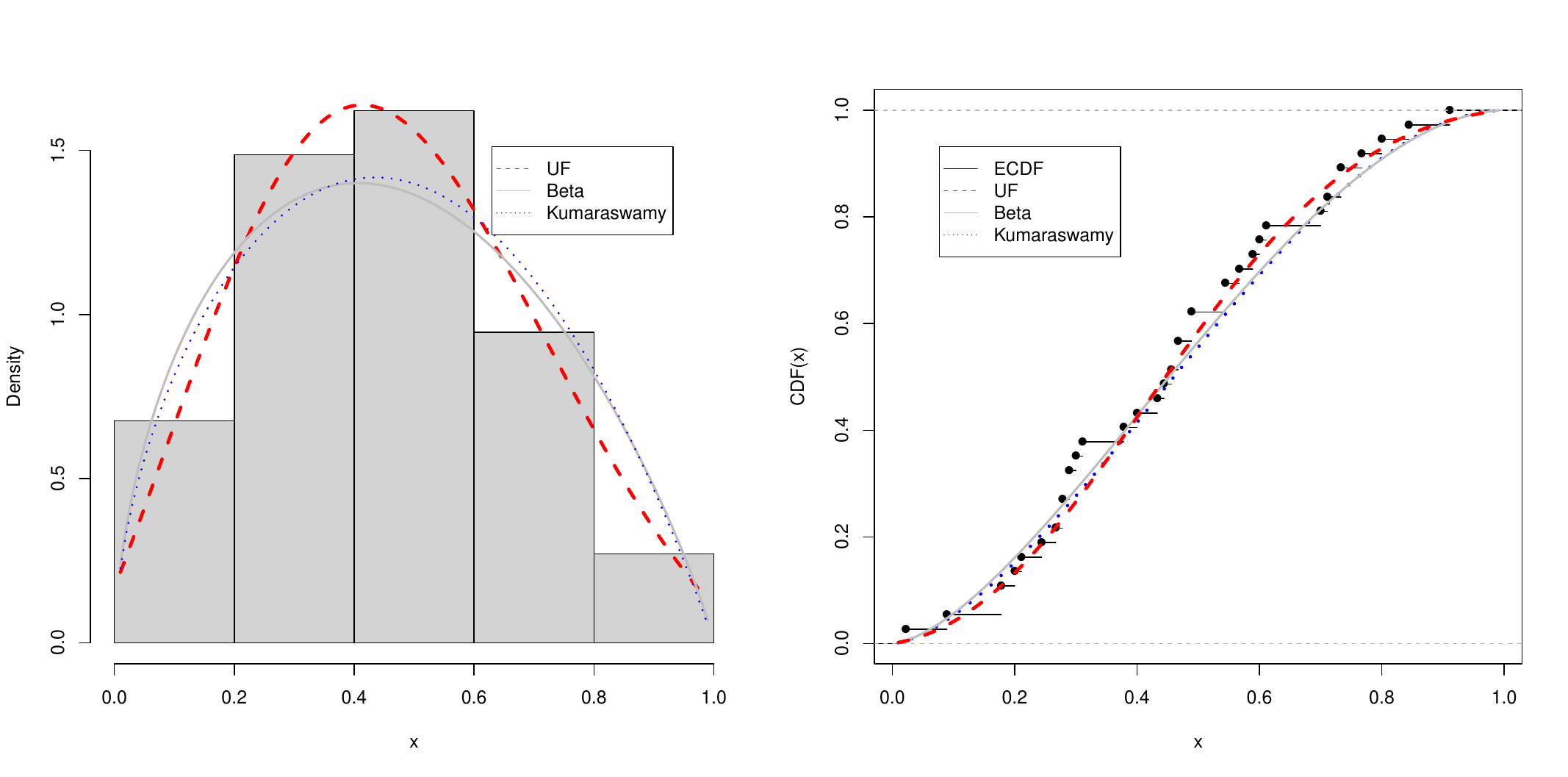}
    \caption{Plots for medium pass completion proportion. On left, histogram and fitted PDFs; on right, empirical CDF and fitted CDFs.} 
    \label{fig:hist_ecdf_fut}
\end{figure}

We use a residuals analysis, to ensure a good model fit. The residuals are calculated using the formula 
$R_i=\hat{F}(w_i) - F(w_i; \hat{\bm \theta})$, where $\hat{F}$ denotes the empirical CDF (ECDF), $\hat{\bm \theta}$ is the estimated parameter, $F(w_i; \hat{\bm \theta})$ is the CDF of the fitted model (unit-Fréchet, Beta or Kumaraswamy) to each observation \(w_i\). 
The Quantile-Quantile (QQ) plots and boxplot of residuals (Figure \ref{fig:uefa_qq}) suggest that the unit-Fréchet model performs as well as the Beta or Kumaraswamy distributions. However, as previously highlighted, analyzing the histogram, ECDF plots, and the corresponding information criteria suggests that the unit-Fréchet model is recommended for a better data fit.

\begin{figure}[htb!]
    \centering
\includegraphics[width=0.9\linewidth]{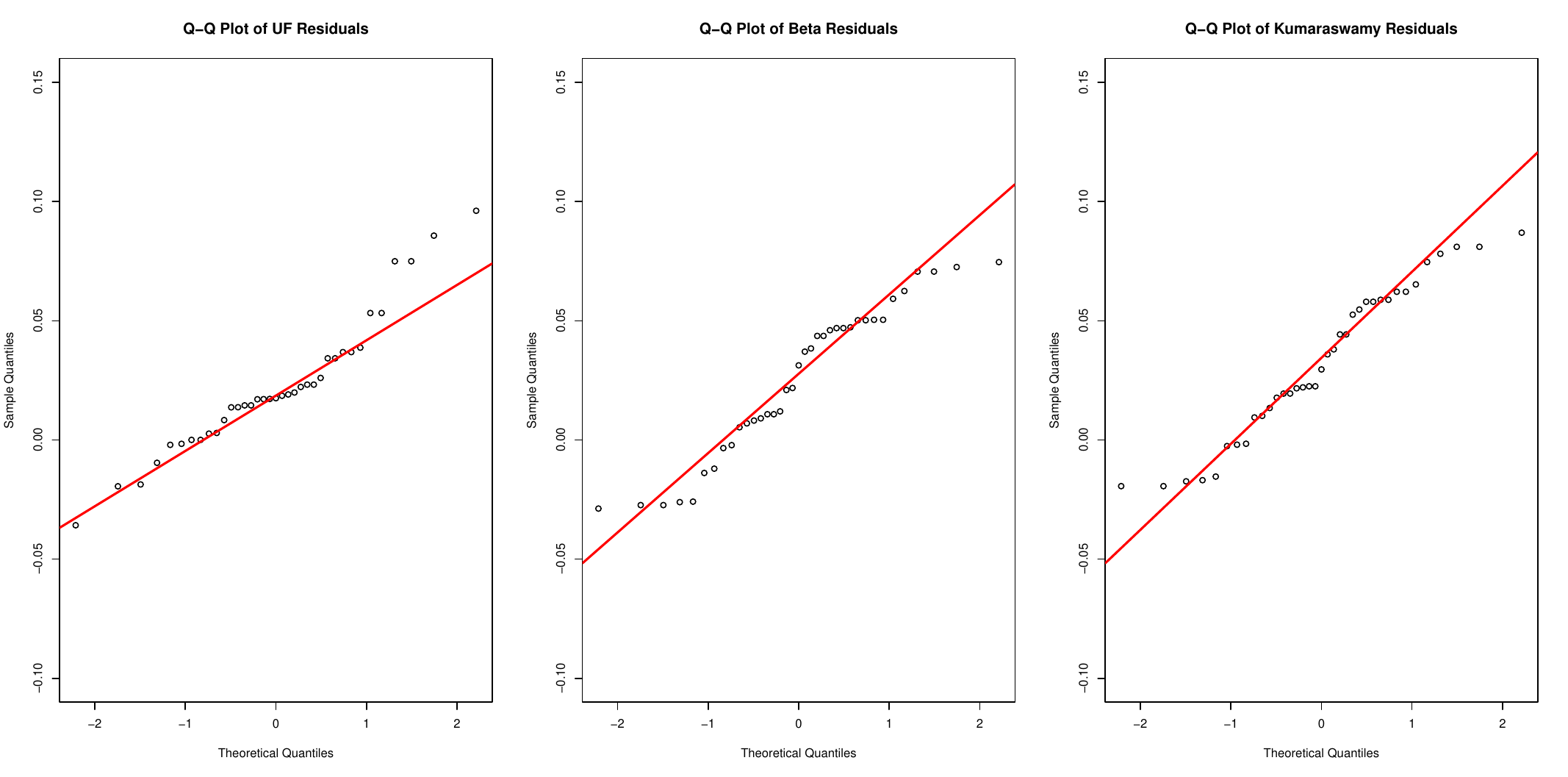}
\includegraphics[width=0.9\linewidth]{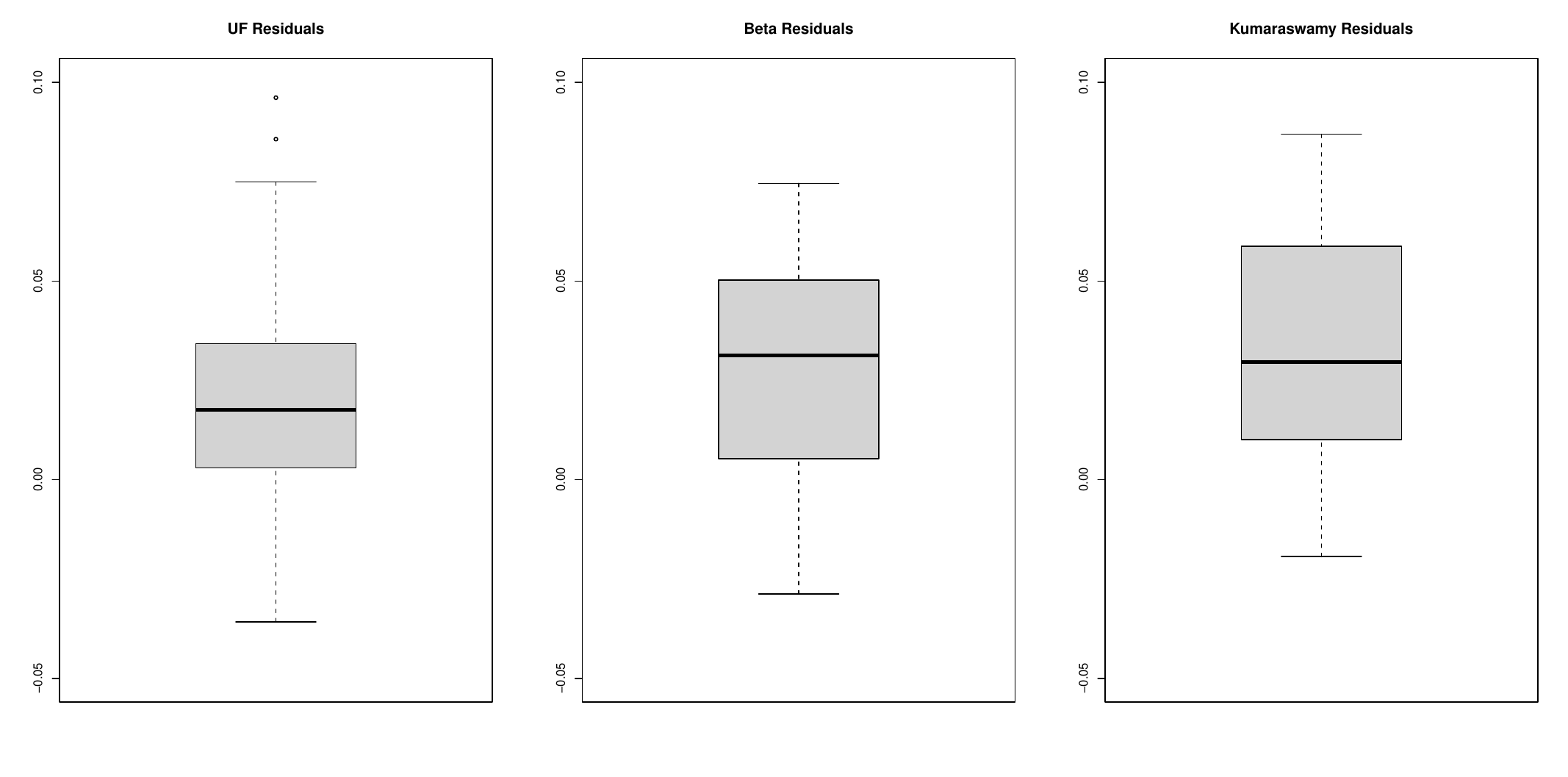}
    \caption{Quantile-Quantile plot (on top) and boxplot (on bottom) displaying residuals from fitted models.}
    \label{fig:uefa_qq}
\end{figure}

\subsection{Income-consumption data}

Recently, \cite{Vila24unitbs} introduced the unit bimodal Birnbaum-Saunders (UBBS) distribution and showed its effectiveness in modeling income and consumption data. Their findings highlighted that the UBBS provided a good fit for such data sets. Here, we show that the unit Fréchet distribution offers an improved fit for modeling these data.

The data \footnote{Available at \url{https://www.bancaditalia.it/statistiche/tematiche/indagini-famiglie-imprese/bilanci-famiglie/documentazione/ricerca/ricerca.html?min_anno_pubblicazione=2008&max_anno_pubblicazione=2008}, (accessed on July 22, 2024)} comes from the Bank of Italy's Survey on Household Income and Wealth from the year 2008. 
We utilized data from the RFAM08 data set (with income as the $Y$ variable) and the RISFAM08 data set (with expenditures as the $C$ variable), as outlined in the Survey on Household Income and Wealth 2008 data description file. We excluded any data entries where income or consumption was negative or unavailable. The final data set analyzed included information from 7,957 families, after removing 20 entries.

These data set was previously modeled in \cite{domma2012stress} and \cite{Lima24}, however, they compared income and consumption through stress-strength measures $\mathbb{P}(X<Y)$.
In this study, we analyze household financial fragility by constructing a new variable $U=Y/(X+Y)$.
In summary, when $X$ and $Y$ represent expenditures and income, respectively,  $U < 1/2$ indicates that the family ended the year with more expenditures than income. If $U > 1/2$, the opposite is true. The case $U = 1/2$ means expenditures and income are equal.

We can highlight that the ratio of the Fréchet random variables model proposed by \cite{NadarajahKotz06} is unsuitable for modeling this data because the random variables $X$ and $Y$ are correlated (Pearson's correlation is 0.717). Conversely, our ratio distribution \eqref{pdf-main} is applicable, as discussed below.

Table \ref{tab:descritiva_receitas_despesas} presents descriptive statistics for $U$. In addition, we note that 81.5\% of the families had $U>1/2$. That means, in general, the families have income greater than expenditures.
Figure \ref{fig:ic_hist_ecdf} shows the fit of the unit-Fréchet, Beta, Kumaraswamy, and UBBS distributions to the density and ECDF of the data. 
The unit-Fréchet model provides the better fit, according to the ECDF plot which agreed with the information criteria (Table \ref{tab:ic_AIC_BIC}).
Analyzing Figure \ref{fig:ic_qqplots}, we conclude that, in general, the residuals of the unit-Fréchet model are more well-behaved than residuals of the other theoretical distributions.

\begin{table}[!ht]
	\centering
	\caption{Descriptive statistics for income-consumption data.}
	\label{tab:descritiva_receitas_despesas}
		\begin{tabular}{crrrrrrrrr}
			\hline
			n & Min. & 1st Qu. & Median & Mean & 3rd Qu. & Max. & Std. dv. & CS & CK\\ 
			\hline
			7,957 &  0.004
			&  0.514
			& 0.553
			&  0.557
			&  0.608
			& 0.935
			& 0.086
			& -0.419
			& 3.343\\
			\hline
	\end{tabular}
\end{table}

\begin{table}[!ht]
	\caption{Estimated parameters for income-consumption data and model selection with AIC and BIC.}
	\label{tab:ic_AIC_BIC}
	\centering
	\begin{tabular}{ccccc}
		\hline
		Model & $\hat{\bm\theta}$ & $ll_{max}$ & AIC & BIC  \\ 
		\hline
        UF &(1.256, 3.779, 0.714)& 8543.067& \bf -17080.13&\bf -17059.19\\
		Beta & (7.492, 6.390) &  6877.40 & -13750.79 & -13736.83 \\ 
        Kumaraswamy & (4.613, 9.608)& 6978.31& -13950.62& -13929.67\\
        UBBS & (0.275, 0.274, 1.041, 1.331, 0.149) & 7062.74 & -14135.49 &  -14170.40 \\ 
		\hline
	\end{tabular}
\end{table}

\begin{figure}[!ht]
	\centering
	\includegraphics[width=1.0\linewidth]{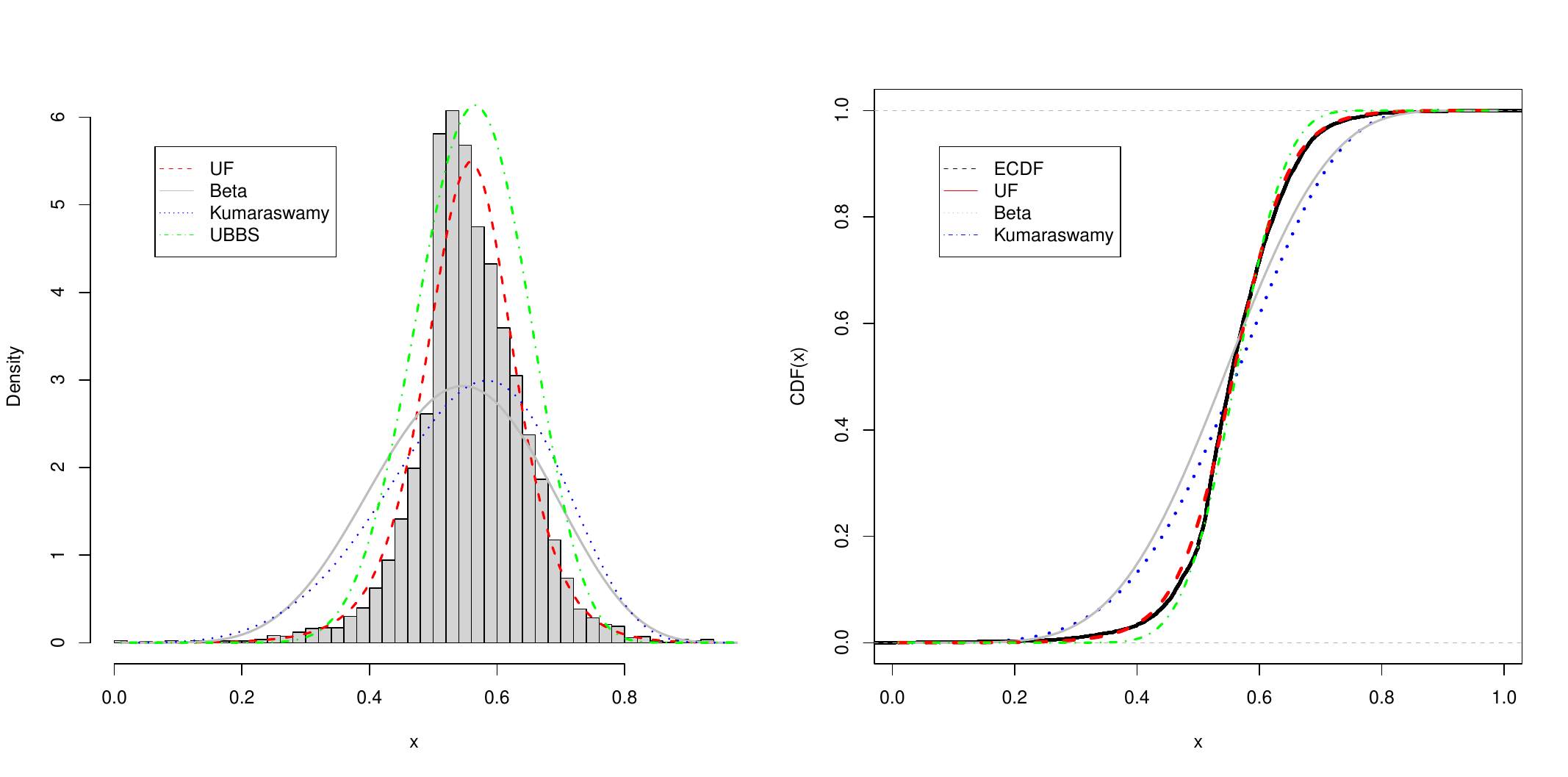}
	\caption{Fitted PDFs (left) and empirical CDFs (right) for income-consumption data.}
	\label{fig:ic_hist_ecdf}
\end{figure}

\begin{figure}[!ht]
	    \centering
	    \includegraphics[width=1.0\linewidth]{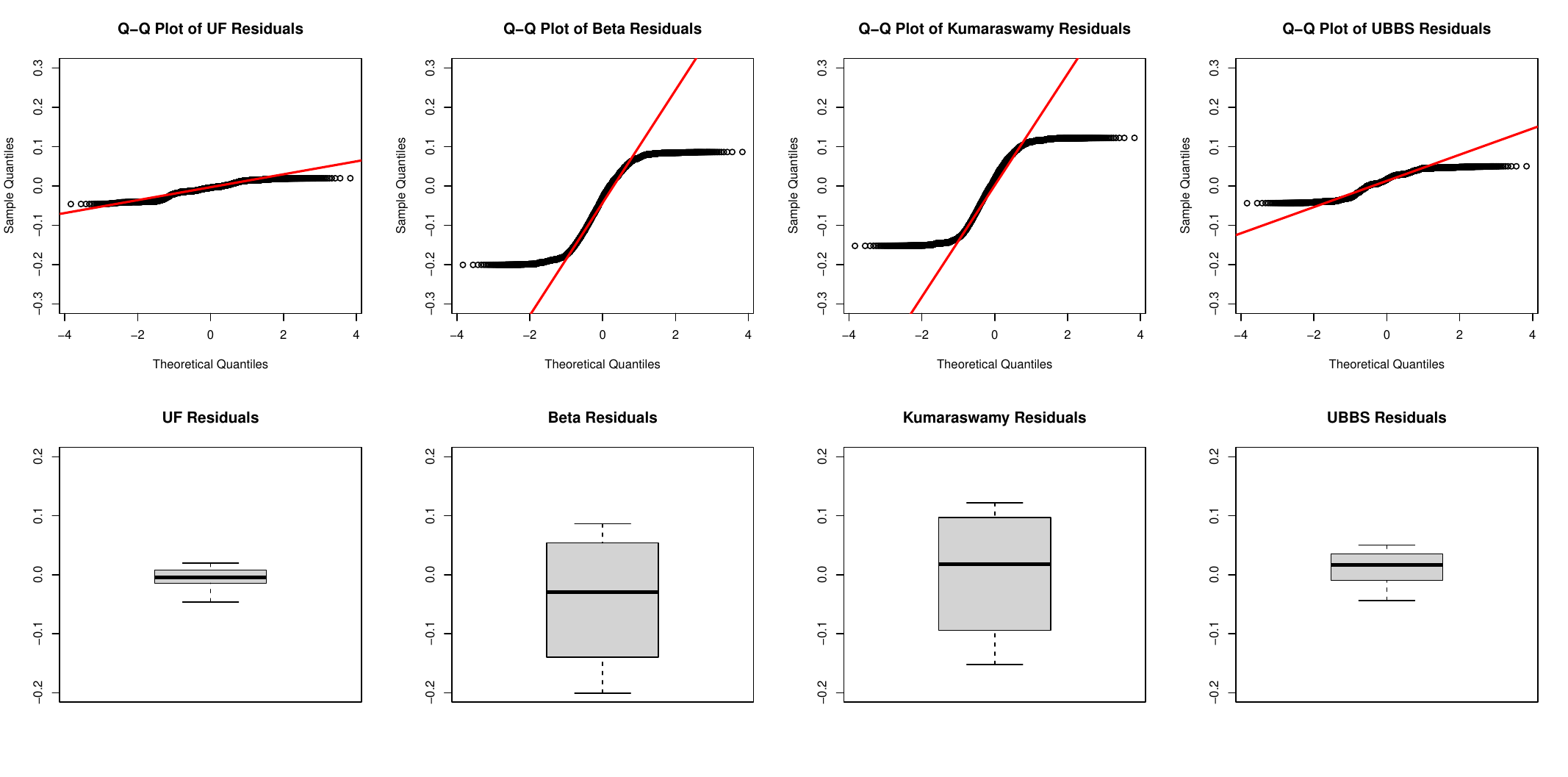}
	    \caption{Quantile-Quantile plot (on top) and boxplot (on bottom) displaying residuals from fitted models for income-consumption data.}
	    \label{fig:ic_qqplots}
	\end{figure}

\newpage

\section{Conclusions}
This paper introduces a novel unit distribution, exploring its generation through ratios of dependent Fréchet variables. Theoretical aspects, including identifiability, symmetry, stochastic representation, characterization, moments, stress-strength probability, quantiles, and maximum likelihood estimation, are rigorously examined. Practical applications are demonstrated via Monte Carlo simulations, football data analysis, and income-consumption modeling.

Potential extensions include analyzing modality, entropy, and shape characteristics (skewness, kurtosis) using moment approximations in Subsection \ref{Approximations for the moments}. Future research can investigate alternative construction principles for unit distributions using quantiles of positive-support random variables, potentially yielding new distributional families.

%
	\paragraph*{Data availability statement}
In this manuscript the main results are theoretical.
Data availability is not applicable to this article as no new data were created or analyzed in this study. 
	\paragraph*{Acknowledgements}
The research was supported in part by CNPq and CAPES grants from the Brazilian government.
	
	\paragraph*{Disclosure statement}
	There are no conflicts of interest to disclose.

\end{document}